 \documentclass[reqno,11pt,a4paper]{amsart}

\pdfoutput=1

\usepackage{amsfonts,amsmath,amssymb}
\usepackage[latin1]{inputenc}
\usepackage{dsfont}
\usepackage{enumerate}
\usepackage{xcolor}
\usepackage[all]{xy}
\def\notshow#1\notshowend{} %
\usepackage{hyphenat}

\usepackage{graphicx}

\usepackage{tikz}
\usepackage{tikz-3dplot}
\usepackage{pgfplots}

\usepackage{multirow}
\usepackage{array}

\usepackage{amssymb}
\usepackage{amsfonts}
\usepackage{mathrsfs}
\usepackage{hyperref}

\usepackage[normalem]{ulem} 

\def\br#1\er{{#1}} %
\def\bb#1\eb{\textcolor{blue}{#1}} 
\def\bm#1\em{\textcolor{magenta}{#1}} %

\newcommand{\ca}{op} 
\newcommand{\cc}{cl} 

  \newcommand{\N}{\mathds{N}}     
  \newcommand{\R}{\mathds{R}}     
  \newcommand{\Z}{\mathds{Z}}     
 \newcommand{\C}{\mathcal{C}}
    \newcommand{\Lo}{\mathds{L}}     

\newtheorem{thm}{Theorem}[section]
\newtheorem{prop}[thm]{Proposition}
\newtheorem{lemma}[thm]{Lemma}

\theoremstyle{definition}
\newtheorem{defi}[thm]{Definition}

\newtheorem{example}[thm]{Example}

\newtheorem{rem}[thm]{Remark}

\newcommand{\be}{\begin{equation}}
\newcommand{\ee}{\end{equation}}
\newcommand{\ben}{\begin{enumerate}}
\newcommand{\een}{\end{enumerate}}
\newcommand{\bit}{\begin{itemize}}
\newcommand{\eit}{\end{itemize}}
\newcommand{\edoc}{\end{document}}

\newcommand{\bq}{\begin{quote}}
\newcommand{\eq}{\end{quote}}

\def\m{\mathcal{M}}

\newcommand{\gl}{\langle \cdot,\cdot \rangle_1}
\newcommand{\gs}{g_{S^{n-1}}}

\newcommand{\g}{g_M}
\newcommand{\go}{g_{\hbox{\tiny{op}}}}
\newcommand{\gc}{g_{\hbox{\tiny{cl}}}}
\newcommand{\ttt}{t}
\newcommand{\parti}{\partial_i} 

\newcommand{\ak}{\hbox{arc}T_k}

\title[]{Globally hyperbolic spacetimes: \\ slicings, boundaries and  counterexamples}

\author[M. S\'anchez]{Miguel S\'anchez} \address{IMAG \&  Departamento de Geometr\'{\i}a y Topolog\'{\i}a, Facultad de Ciencias, \hfill\break\indent Universidad de Granada,\hfill\break\indent Campus Fuentenueva s/n, \hfill\break\indent 18071 Granada, Spain}\email{sanchezm@ugr.es}

\begin{document}

\begin{abstract} 

 The Cauchy slicings for globally hyperbolic spacetimes and their relation with the   causal boundary    are surveyed and revisited, starting at  the seminal conformal boundary constructions by  R. Penrose. 
Our study covers: (1) adaptive possibilities and techniques  for their Cauchy slicings,  (2)      global hyperbolicity of sliced spacetimes, 
(3)  critical review on the conformal and causal boundaries for a globally hyperbolic spacetime, and (4) procedures to compute the causal boundary of a Cauchy temporal splitting  by using isocausal comparison with a static product.
New simple counterexamples on $\R^2$  illustrate a variety of possibilities related to these splittings, such as the logical independence (for normalized sliced spacetimes) between the completeness of the slices and global hyperbolicity, the necessity of uniform bounds on the slicings in order to ensure global hyperbolicity, or the insufficience of these bounds for the computation of the causal boundary. A refinement of one of these examples shows that  the space of all the (normalized, conformal classes of) globally hyperbolic 
metrics on a smooth product manifold $\R\times S$ is not convex, even though it is path connected by means of piecewise convex combinations.


\vspace{5mm}



\noindent {\em Note:} 
The results here were explained at the online meeting SCRI21 ``Singularity theorems, causality, and all that.
A tribute to Roger Penrose'' June 14-18 (2021),
which  is available at its website
\url{https://sites.google.com/unifi.it/scri21}.

\end{abstract}
\maketitle

\newpage

\tableofcontents

\newpage

\section{Introduction}\label{s_Intro} 

Among  the many celebrated contributions of Roger Penrose to Mathematical Relativity, the introduction of the conformal boundary \cite{Pe} and, in collaboration with Geroch and Kronheimer, the causal one \cite{GKP}, have been specially fruitful. The former underlies notions such as asymptotic flatness and conformal infinity, which are commonly used in a huge variety of relativistic topics. The latter provides a general intrinsic  boundary for 
spacetimes, 
developed later by many authors, which has been linked with other outstanding constructions in Differential Geometry, such as the Busemann and Gromov boundaries. In this contribution, we will focus on the computability of these boundaries for arbitrary globally hyperbolic spacetimes. 
 This class of spacetimes  underlies other Penrosian concepts such as cosmic censorship.     
Their structure and some of their properties will also receive special attention here. We will make a critical review on these topics linking different techniques and providing new consequences. Moreover, we construct some counterexamples about a variety of issues, which show the accuracy of the known techniques and may orientate further developments. The style is intended  pedagogical, so that the reader can delve into the literature and apply the techniques in different situations.

A remarkable result by Geroch \cite{Ge} ensures that a globally hyperbolic spacetime admits an acausal Cauchy hypersurface and, then, it splits topologically as $\R\times S$. Starting at this landmark, a series of results in 2003-11 \cite{BS03,BS05,BS07,MS}  improved the splitting adding quite a few general possibilities. Among them,  such spacetimes can be written globally as a parametrized orthogonal product 
\begin{equation}\label{e_00}
\R\times S, \qquad g_{(t,x)}=-\Lambda(t,x) dt^2+(g_t)_{x},
\end{equation}
 with each slice  $t=$constant a Cauchy hypersurface (see details around \eqref{e_orthgonalsplitting}). 
The existence of this splitting poses some issues. First to identify, among the spacetimes  admitting an expression such as \eqref{e_00}, which of them are globally hyperbolic; this will yield additional information on the space of all the globally hyperbolic spacetimes. Then, to  use the globally hyperbolic structure \eqref{e_00} in order to obtain information of the spacetime from the properties of $\Lambda, g_t$ or, in the case of conformally invariant properties (as those studied in the present paper), from the quotient $g_t/\Lambda$. 

The first part of the article focuses on the structure of globally hyperbolic spacetimes (\S \ref{s3} and \S \ref{s4}) and the second one  on the causal boundary (\S \ref{s2a}  and \S \ref{s_causalboun_for_sliced}).  The last section \S \ref{s6} develops  detailedly the main counterexamples which illustrate the results along the paper. 

In \S \ref{s3}, a review on the possibilities of the globally hyperbolic splittings and the underlying techniques is carried out. The three first subsections explain briefly the ideas in the aforementioned series of results. Here, the so-called folk problems of smoothability serve as a historical and scientific guide to understand the flexibility and increasing adaptability of the splittings; the analytic case 
(which did not appear in previous references, as far as the author knows) 
is also considered in \S \ref{ss_analytic}. 
In the last two subsections, \S \ref{ss_dynamical}, \S \ref{s_Festability},  we  compare these techniques with  more recent approaches. 
The latter have extended the scope and links of the problems to   other fields such as Dynamical Systems, and they emphasize the interest of considering   cone structures more general than the Lorentzian ones. So,  we introduce more general cone structures by using Finslerian elements, \S \ref{s243},
and explain how the original techniques  can also be  used to study them, posing new questions in Lorentz-Finsler Geometry,  \S \ref{s244},  \S \ref{s245}.

In \S \ref{s4}, we study when the temporal function $t$ in \eqref{e_00} (under the conformal choice $\Lambda\equiv 1$) has Cauchy slices. First we discuss the role of the completeness of the metrics $g_t$'s in the  slices. We emphasize that  this condition is neither  necessary nor  sufficient  for global hyperbolicity, as shown by explicit counterexamples (\S\ref{ss0} and \S \ref{ss1}). However, we show that  a uniform bound with respect to a complete Riemannian metric imply that $t$ is a Cauchy temporal function (Proposition \ref{p_Bari}). The necessity of such bounds is stressed by means of a fine tuned counterexample (Remark \ref{r_ej3}, \S \ref{ss2}).  
What is more,  the structure of the space of globally hyperbolic metrics on a product $\R\times S$ is  better understood then, \S \ref{s_paracausal}. Indeed, the existence in \eqref{e_00} of a normalized representative of each globally hyperbolic conformal class poses the question whether this space is {\em convex}. However, a further  counterexample (\S \ref{ss3}) gives a negative answer. This completes a result by Moretti et al. in \cite{MMV}, who introduced  paracausal transformations for this issue (see Theorem \ref{t_paracausal}, Remark \ref{r_paracausal}).  

In \S \ref{s2a} we review both, the conformal and causal boundaries. The conformal one is studied in the first two subsections, \S \ref{s2.1}, \S \ref{s2.2}. We  stress specially the limitations of this boundary  which lead to the causal  one. The causal boundary is studied in the other two subsections, being the first aim to explain the subtleties of its definition, \S \ref{s2b1}, and the second one to describe the causal boundary of a static product (which will serve as the starting point for the boundary of more general globally hyperbolic spacetimes), \S \ref{s_cbounc_staticproduct}.
 
 Specifically, after  constructing Penrose-type conformal embeddings into static products $\R\times M_k$ of models spaces $M_k$,    we discuss the strictly causal  elements of the conformal boundary, \S \ref{ss2.1}, as well as the requirements to be fulfilled by an open conformal embedding $i: \m \hookrightarrow \hat \m$ in order to give a meaningful conformal boundary, \S \ref{ss2.2}. In particular, we emphasize the role of {\em chronological completeness}, a hypothesis weaker than the precompactness of $i(\m )$, which provides a suitable characterization of global hyperbolicity (Proposition \ref{p_timelikepoint_conf}).  However, in \S \ref{ss2.3} we point out that,  in general, a spacetime $\m$ will not admit any 
 conformal boundary, even if $\m$ is globally hyperbolic and the requirements for this boundary are weaker than the usual ones. Indeed, Example \ref{ex_NO_confbound}, considers  the aforementioned products $\R\times M_k$, which  admit a conformal boundary with two lightcones. One would expect an analogous boundary if $M_k$ is replaced by any Cartan-Hadamard manifold (as it will occur for the causal boundary). However, the Weyl tensor (and, thus, $i(\m )$ for any conformal embedding)  may not admit any continuous extension.
 
 About the notion of causal boundary $\partial_c\m$ for a spacetime $\m$, the issues of the original definition
 are briefly explained in \S \ref{s2b1}, with special focus on the globally hyperbolic case. These issues involve the  definition of $\partial_c\m$ as a pointset  (which, in particular, yields a characterization of global hyperbolicity Proposition \ref{p_timelikepoint_caus}), 
 \S \ref{s_Pairings}, as a causally ordered set, \S \ref{s_Causality} and, specially, as a topological space.      Following \cite{FHS11}, we consider the chr-topology (which goes back to Harris' work \cite{Ha00}) and point out some implicit mathematical subtleties in its definition,  \S \ref{s_Topology}.
In \S \ref{s_cbounc_staticproduct}, the causal boundary of a static product $\R\times M$ is explained; this will be  essential for the ``isocausal comparison'' later. A remarkable fact is that this boundary is connected with the classical  
Gromov boundary of a Riemannian manifold (through an intermediate Busemann boundary). This was analyzed systematically in a  more general setting, see \cite{FHS_Memo}, and a brief intuitive idea is given here. More precisely, the results for our case are summarized first,  \S \ref{s_Summary caus bound}, then, the meaning of this summary is explained with an explicit example, \S \ref{s_brief_expalantion} and an application to Generalized Robertson-Walker spacetimes is also included, \S \ref{s_GRW}. Summing up, the causal boundary for static products comes out completely well-understood and the possible appearance of ``bad'' properties  in  mathematically taylored examples becomes harmless. 

In \S \ref{s_causalboun_for_sliced}, we study how to compute the boundary of a globally hyperbolic spacetime $\m$ considering again a conformal representative (as in \eqref{e_00} with $\Lambda\equiv 1$) and trying to compare it with the boundary of a static product $\m_0=(\R\times S,g_M=-dt^2+g_M)$. This would be analog to the procedure in $\S \ref{s_unifbound}$ which ensures global hyperbolicity by comparing with static products, but a significative  difference appears now: even uniform bounds type $c_1g_M<g_t<c_2g_M$ are {\em not} enough to ensure the identification of $\partial_c\m$ with $\partial_c\m_0$ (see \S \ref{s_lowering} and the counterexample in \S \ref{ss4}).   
As discussed in \S \ref{s_isocausal}, the difficulty appears because the causal boundary is one of the causal properties {\em not} preserved by the so-called {\em isocausal} relation, introduced by Garc\'{\i}a-Parrado and Senovilla in \cite{GPS}. 
Noticeably, a simple strenghtening of such a uniform bound (the hypotheses in \eqref{e_alpha_bound}, \eqref{e_alph_asymp}) will suffice to obtain an isocausal comparison with very relevant information. Indeed, applying the general results in \cite{FHS13}, 
we show how $\partial_c\m$ {\em contains} $\partial_c\m_0$ and, what is more, $\partial_c\m$ can be regarded as a ``strained'' version of $\partial_c\m_0$, see
\S \ref{s_computecbound}. 

As a step forward, in \S \ref{s_ghwithboundary} we consider the case of globally hyperbolic spacetimes-with-timelike-boundary $\overline{\m}$. Even though the interior $\m$ of such a spacetime cannot be globally hyperbolic, the identification of the timelike boundary $\partial \m$ with the nontrivial pairs of $\partial_c\m$ (see Proposition~\ref{p_timelikepoint_borde}) provides a fruitful combination of the causal and conformal boundaries.  After some general properties 
in \S \ref{s_borde1}, in \S \ref{s_borde2} we focus on the extension of the Cauchy orthogonal splitting \eqref{e_00}. Here, the key point is that $\nabla \tau$ must remain tangent to $\partial \m$ at all the points of the boundary  (see Theorem \ref{t_princ_borde}). Moreover, we also explain how the other issues studied above  can be extended to the case with boundary (Remark \ref{r_withbound_extensionrepera}). In the author's opinion, the adaptability of these spacetimes to different situations as well as the variety of available tools (as those explained here), may make them very atractive for future research.

About the counterexamples in \S \ref{s6}, it is worth pointing out that all of them are  type $(\mathds R ^2,g=-dt^2$ 
$+f^2(t,x)dx^2)$, which underlies their robustness. In particular, they are  smooth and  can be chosen analytic when convenient (see Remark \ref{r0}). A first simple counterexample (implicit in \cite{S})    
shows non-global hyperbolicity 
when all the $t$-slices are  complete, \S \ref{ss1}; other example with one slice  incomplete but 
all the slices Cauchy will be also achieved \S \ref{ss0}. Starting at that first example,  the  global hyperbolicity of such a $g$ is shown not to imply global hyperbolicity for   $-dt^2+(f^2(t,x)/2)dx^2$, stressing so the type of appropriate bounds for this property\footnote{In addition, a technical property which might have interest in its own right is proved:  for any Lorentz metric on $\R^2$, all the the spacelike curves are acausal, see footnote \ref{foot}.}, \S \ref{ss2}. A further modification shows two metrics $g,g'$ which are globally hyperbolic but $g+g'$ is not, \S \ref{ss3}, implying consequences for the non-convexity of the space of globally hyperbolic metrics.  Finally,    globally hyperbolic isocausal metrics $g,g'$ with non-isomorphic  causal boundaries are exhibited, \S \ref{ss4}. 

As a last remark, the study of general spacetimes (globally hyperbolic or not) which admit a sliced expression as in \eqref{e_00} but adding  cross terms between the $\R$ and $S$ parts, has an obvious interest and applicability. It is worth emphasizing that  this has already been carried out and, indeed, many results above are obtained by tuning such general results. Noticeably, the stationary case 
(that is, when all the elements  are $t$-independent, so that $\partial_t$ becomes a timelike Killing vector field) leads to a specific Finsler metric of Randers type, which plays the role of $g_t/\Lambda$; this 
metric characterizes the causality of the spacetime \cite{CJSa} as well as the causal boundary $\partial_c \m$ \cite{FHS_Memo}. Again, this time-independent case can be used
for isocausal 
comparison  when studying the time dependent one,   \cite{FHS13}. Moreover, the approach can be even extended to the case that $\partial_t$ is any Killing vector field transverse to the slices of $t$ (as happens in physical models beyond a Killing horizon), 
\cite{CJSb, JS}. Such possibilities  are briefly commented in Remarks \ref{r_linksLF0}, \ref{r_linksLF1}, \ref{r_linksLF2},  \ref{r_withbound_extensionrepera}, and the reader is referred to the survey \cite{JPS} for an account of   Lorentz/Finsler links. Notice, however, that these links appear in standard Relativity, that is,  one is not generalizing the cones structures (as done in  \S \ref{ss_dynamical}) at this step. Anyway, one can go a step forward and use properly Lorentz-Finsler geometry for applications in both, extensions of Relativity and Classical Mechanics, \cite{JPS21}.

\section{Structure  of globally hyp. spacetimes}\label{s3}


In general, we will follow conventions as in \cite{MinSan}; in particular, the 0 vector is not regarded as causal.

\subsection{Geroch's results} In a celebrated article, Geroch (1970)  \cite{Ge} proved the following characterizations of  global hyperbolicity (defined classically as being strongly causal  with compact diamonds
 $J(p,q):=J^+(p)\cap J^-(q)$, $p,q\in\m$) for  a spacetime  
$\m$:
(i) $\m$ admits an acausal, topological  Cauchy hypersurface $S$ and (ii)  $\m$ admits a Cauchy time function $t$, i.e.,   $t$ is a continuous function which  increases strictly on future-directed causal curves and its slices $t=$ constant are (acausal, topological)  Cauchy hypersurfaces. Moreover, he also proved that, in this case, $\m$ becomes homeomorphic to $\R\times S$ and all the Cauchy hypersurfaces (acausal or not) are homeomorphic. It is worth pointing out that a Cauchy hypersurface, defined as a subset $S$ which is crossed exactly once by any inextendible timelike curve, must be a (just achronal) topological locally Lipschitz hypersurface and, moreover, its existence also implies global hyperbolicity (see the detailed proofs in \cite{O} as well as other Cauchy criteria in \cite{Gal1}).

Even though the details of the proof are demanding, the construction of the Cauchy time function $t$ is simple and beautiful.  Just take any {\em admissible} measure $\mu$ on $\m$ (the one associated with any auxiliary Riemannian metric such that $\mu(M)<\infty$ is enough, see \cite{Dieckmann}) and the job is made by the function
\begin{equation}\label{e_Geroch}
t(p):=\log \left(\frac{\mu(J^-(p))}{\mu(J^+(p))}\right), \qquad \forall p\in\m .
\end{equation}

This is a key result in Causality, even if any conclusion on  the smooth or, more restrictively,  spacelike 
character of the Cauchy hypersurfaces is conspicuously absent. Indeed, Geroch gave also an additional reasoning to ensure the stability of global hyperbolicity in the set of all the Lorentz metrics on $\m$. However, this could have been circumvented if the Cauchy time function $t$ were {\em smooth with spacelike} slices i.e., {\em temporal} (in the sense  below), as we will  discuss in~\S \ref{s_Festability}.

\subsection{Folk problems} This issue  was studied in Seifert's thesis\footnote{Defended in 1968 and cited in the classical  Hawking and Ellis book (1973) \cite{HE}; the concerned results were published later, in 1977 
\cite{Se77}.}, however, the proof was unclear. Indeed, Sachs and Wu coined the term {\em folk theorem} for the simplest underlying question, namely, whether   
a globally hyperbolic spacetime admits a (smooth) spacelike hypersurface   \cite[p. 1155]{SW}; this was taken into account in  the posterior mathematically careful 
bibliography\footnote{For example, Beem's et al. book cited explicitly Sachs and Wu  in \cite[p. 65]{BEE}. Typically, the existence of      spacelike Cauchy 
hypersurfaces \cite{O,Pe2,Newman} (or temporal functions  \cite{Bartnik, EcHu}) was regarded as an additional hypotheses,     even stressing that spacelike means only acausal in Seifert's terminology \cite{Gal1}.}. This type of problems appear for several causal constructions and some times they are referred as  folk problems on {\em  smoothability}. However, we must emphasize that  the issue remains even in the smooth case: on the one hand,  a smooth acausal Cauchy hypersurface $S$ may be lightlike (i.e. degenerate) at some point and, on the other,  non-smooth Cauchy $S$ can always be obtained as the level of a smooth time function (see the discussion below Theorem~\ref{t_LMP}). What is more, this  is related to stability, because  a lightlike  hypersurface $S$ is always {\em unstable} as a Cauchy one (arbitrarily small compactly supported smooth perturbations {\em do} contain non-achronal ones around a degenerate point) while a  spacelike Cauchy hypersurface will be always $C^0$ stable. 

Taking into account the observations above we will use the nomenclature introduced in \cite{BS05} and say that a function $\tau$ is {\em temporal} if it is smooth (as differentiable as allowed by the smoothness of the spacetime $\m$) with past-directed timelike gradient; in particular, $\tau$ must be a time function. Consistently, if $\tau$ is Cauchy temporal, its levels $\tau$=constant are Cauchy hypersurfaces (necessarily spacelike). 
Summing up, the main folk problems at the end of the seventies were to prove the assertions:

\ben
\item \label{F1} A globally hyperbolic spacetime admits a (smooth) spacelike Cauchy hypersurface (Sachs \& Wu).

\item \label{F2} A globally hyperbolic spacetime admits a  Cauchy temporal function~$\tau$. 

\item \label{F3} If a spacetime admits a time function then it admits a temporal one (completing so the alternative definitions of stable causality). 

\item \label{F4} Any globally hyperbolic spacetime can be isometrically embedded in some Lorentz-Minkowski spacetime $\Lo^{N}$ for large $N\in \N$ (in the spirit of Nash' theorem \cite{Nash}). 

\een

\subsection{Solution of the folk problems}\label{s_folk} These questions were solved systematically in a series of papers in 
starting at \cite{BS03}. 
This reference introduced a basic type of functions $h_p$ (a sort of modified Lorentz distance function) and a way to operate them 
which was used  in  the subsequent constructions.
Roughly speaking, in Riemannian Geometry, global objects
are constructed frequently from local ones by using partitions of the unity. Here, however, the causal character of the gradient for functions in
the partition would be, in principle, uncontrolled. Then, the intended objects will be constructed
by using the paracompactness of $\m$ but avoiding partitions of the unity.

Next, we will make a brief summary of   the original articles \cite{BS03,BS05,BS07,MS}, giving some ideas on the proofs and related issues.  Most details can be seen in book format at Ringstr\"om's \cite[Chapter 11]{Ringstrom}, which is strongly recommended.
It is worth pointing out that  some extra bonuses were obtained in this epoch \cite{MinSan}, among them, significatively,  the simplification of the aforementioned classical definition of global hyperbolicity, namely: the hypothesis {\em strong causality}  can be weakened into 
{\em causality}\footnote{Previously, Beem and Ehrlich (see \cite[Lemma 3.4]{BE79b} or \cite[Lemma 4.29]{BEE}) had noticed that, under strong causality, the compactness of the diamonds $J(p,q)$ can weakened into the compactness of their closures. However, under this hypothesis, strong causality cannot be weakened into causality as shown by a Carter's example of  imprisoning causal spacetime, see \cite[p. 368]{Carter} and \cite[p. 195]{HE}. Recently, Hounnonkpe and Minguzzi \cite{HM} have proven that, in dimension $n\geq 3$, the non-compactness of $\m$ (in addition to the compactness of the diamonds) suffices for the definition.}~\cite{BS07}. 

\subsubsection{Problem \ref{F1}: spacelike Cauchy  and basic ideas}\label{s_BS03}  Let us see  the ideas 
to construct a spacelike Cauchy hypersurface $S$.  
 Consider a topological Cauchy hypersurface $S_0$, say, the slice $t=0$  of Geroch's function \eqref{e_Geroch}, and $p\in S_0$. We will use  functions  $h_p\geq 0$, normalized  so that $h_p(p)=1$   with  support in  a small compact neighborhood of $p$  satisfying that, whenever $q\in J^-(S_0)$ and $h_p(q)\neq 0$, then its gradient $\nabla h_p (q)$ is timelike and past pointing. 

Such a function can be easily constructed by starting at some close $p'\ll p$ in any small convex neighborhood $C_p$ of $p$ 
and taking  the  Lorentzian distance function $d_{p'}$ (or {\em 
time-separation}) from $p'$  in $U$, that is,  $d_{p'}(q)$ is 
 the length of the unique  future-directed causal 
pregeodesic  from $p'$ to $q$ entirely contained in $C_p$ if 
this curve exists, or 0 otherwise. Now, modify $d_{p'}$ as follows:  
(i) as $d_{p'}$ is smooth in $C_p$ except when vanishes, take  
$\exp({-1/d_{p'}^2})$ (and normalize it at $p$ 
multiplying by a constant), 
and  (ii)  this function satisfies the required properties except that it does not match continuously with the value 0 in 
$I^+(S_0)\setminus C_p$, so, modify it therein (say, multiply it by a suitable bump function). 

 
 The technique to obtain the required $S$ consist in
 taking a second level  of Geroch's function $S_-=t^{-1}(-1)$
 and choose $S$ as any closed, connected, smooth spacelike   hypersurface  in  $I^+(S_{-})\cap I^-(S_0)$ (then, $S$ will be necessarily Cauchy, see \cite[Corollary 11]{BS03}).
 With this aim, one constructs a function $h$ satisfying  (i) $h\geq 1$ on $S_0$, (ii) $h=0$ on $I^-(S_-)$ and (iii) the gradient of $h$ is timelike and past-pointing on the open subset $(h^{-1}(0,1))\cap I^-(S_0)$. Indeed,  the required $S$ can be taken then as $(h^{-1}(1/2))\cap I^-(S_0)$. 
 
 The function $h$ can be  constructed  by using paracompactness, taking a  locally finite sum  of  the functions $h_p$ above
 (the condition (iii) will hold in the sum due to the convexity of the causal cones), \cite[Prop. 14]{BS03}. 

\begin{rem}\label{rem_oldnew}
Even though we have proven the existence of a spacelike Cauchy hypersurface,  just the  existence of a smooth one implies the strengthening of  Geroch's homeomorphism $\m \cong \R\times S$ into a diffeomorphism. This has relevant consequences \cite{CherNem13, Torres} and must  be taken into account in  old topological results such as \cite{ClarkeNewman, Newman}. 

For example,  Chernov and Nemirovski \cite{CherNem13} proved that 
every contractible
smooth 4-manifold admitting a globally hyperbolic Lorentz metric is diffeomorphic to
the standard $\R^4$. In particular, the celebrated exotic differentiable structures on $\R^4$ (coming from classical works by Freedman, Taubes and others) cannot admit globally hyperbolic metrics. 
\end{rem}

\subsubsection{Problem \ref{F2}:
Cauchy temporal functions and 
beyond.} \label{s_CauchyTemporal_construction} A Cauchy temporal function $\tau$ 
was constructed in \cite{BS05} by introducing the following modification in the previous technique. 
Recall that, up to now, we have not been worried about the causal character of $\nabla h$ beyond $S_0$. On the contrary, now we are 
interested in the construction of a {\em temporal step function} $\tau_0$ around $S_0$. Namely, choosing  $t_-<t_a<0<t_b<t_+$,  such a 
function must satisfy: (i)~$-1\leq \tau_0
\leq 1$, (ii) $\tau_0\equiv -1$ in $J^-(S_{t_-})$,   $\tau_0\equiv 1$ in $J^+(S_{t_+})$, (iii)~$\nabla \tau_0$ is either  past-
directed timelike  or 0 everywhere, and (iv) $\nabla \tau_0$ does not vanish on the strip 
of Geroch's Cauchy hypersurfaces $S_{t'}$ with $t_a< t' < t_b$. Indeed, once an analogous temporal step function $\tau_k$, $k\in\Z$, is constructed for each  Geroch slice $S_k$ (with a suitable choice of  $t_-, t_a, t_b, t_+$)
the required Cauchy temporal $\tau$ appears as a sum\footnote{The ordering in this sum makes it finite at each point. Indeed, by (ii)  each parenthesis will vanish on a sequence of  Geroch's strips centered at $t=0$ and exhausting $\m$.} $\tau= \tau_0 +\sum_{k\geq 0} (\tau_k+ \tau_{-k})$, 
(see around \cite[formula (2)]{BS05} or \cite[formula (11.4)]{Ringstrom}).

About the construction of the step temporal $\tau_0$, notice first that if (iv) were required only for 
the slice $t'=0$, then $\tau_0= -1+2h^+/(h^+-h^-)$ would work; here  $h^+\geq 0$ is essentially just the function   $h$ in \S \ref{s_BS03} and 
$h^-\leq 0$ a dual one towards the past. Moreover, if $K$ is any compact set
included in $t^{-1}(t_-,t_+)$ then
  one can take an averaged sum $\tau_K$ of such step temporal functions for a finite number of slices so that $\nabla \tau_K$ will be timelike on $K$ (and (i)---(iii) still hold).
Now, to ensure (iv) for $t'$ in the whole interval $(t_a,t_b)$, one takes an increasing sequence of compact 
subsets  $K_j$ such that 
$
[t_a, t_b] \times S 
\subset \cup_j  K_j 
\subset (t_-, t_+) \times S
$, 
and construct such a function $\tau_{K_j}$ for each $K_j$ as above. Then, consider the series
  $\sum_j \tau_{K_j}/C_j$ where the  constants $C_j$'s are chosen so that, for $j\geq j_0$,  the series  as well as of all the derivatives up to order $j_0$ are uniformly convergent in $K_{j_0}$. Then, the sum is smooth and $\tau_0$ is obtained by normalizing it (by composing with a suitable non-decreasing function) to ensure that (i) and (ii) holds.

It is worth highlighting that the properties of $\tau$ yield  a  {\em global orthogonal splitting}. Indeed, moving one slice $S$ of $\tau$ by the flow of the timelike vector  field $T=-\nabla \tau/|\nabla \tau|^2$ one writes the globally hyperbolic spacetime  as a {\em Cauchy temporal splitting}:
\begin{equation}\label{e_orthgonalsplitting} 
\m= \R\times S, \qquad g= -\Lambda d\tau^2+g_\tau,
\end{equation} where $\tau$ is identified with the natural   projection $ \R\times S \rightarrow \R$, the {\em lapse} $\Lambda=1/\nabla\tau|^2 (>0)$  is a function on the whole $\m$ and (abusing of  notation) $g_\tau$ is a Riemannian metric on each Cauchy slice at constant~$\tau$, \cite[Prop. 2.4]{BS05}. That is, there are no cross terms ({\em shift}) between the $\R$ and $S$ parts and, for any coordinates $x^i, i=1,\dots n$ on $S$ one can choose $x^0=\tau$ so that: 
$$g_{00}(\tau,x^i)=-\Lambda(\tau,x^i),\quad  g_{j0}(\tau,x^i)=0,\quad  g_{jk}(\tau,x^i)=(g_\tau)_{jk}(x^i).$$  

A further result on the  consistency of Cauchy hypersurfaces and  temporal functions (with direct applications to the uniqueness of solutions to  wave equations \cite[Ch. 12]{Ringstrom}, including its quantization \cite{BGP}) is the following. 
 
\begin{thm}\cite[Thm. 1.2]{BS06} \label{t_LMP}
Any spacelike Cauchy hypersurface $S$ can be regarded as the slice $\tau=0$ of a Cauchy temporal splitting as  \eqref{e_orthgonalsplitting}.
\end{thm}
This  is proven in two steps. First, one realizes that any acausal Cauchy hypersurface $S$ (even if it is only continuous) can be obtained as the level $t=0$ of a {\em smooth} Cauchy time function\footnote{The reader should be cautious when looking at the bibliography because, in spite of this result, some authors still use the name ``smooth Cauchy time function'' to mean what we call ``Cauchy temporal''.} $t$. Indeed, $t$ can be chosen with $\nabla t$ past-directed timelike everywhere but on $S$ (notice that it must vanish necessarily on $S$ if, for example, $S$ is nowhere differentiable).
 This is carried out taking  temporal functions $\tau^\pm$ on $I^\pm(S)$ (which are globally hyperbolic open subsets) and merging carefully $\exp(\pm \tau^\pm)$ at $S$, see 
 \cite[Th. 1.15]{BS06}. 
 Once this  ensured, when $S$ is spacelike we can take a step temporal function $\tau_0$ around $S=S_0$, and 
the required function is  $\tau=t+\tau_0$, \cite[Th. 1.2]{BS06}.

Finally, let us point out another  issue in the classical relativistic setting with natural applications to quantization \cite{BR}  (in the framework of locally covariant quantum field theories \cite{BFV}).

\begin{thm}\cite[Thm 1.1]{BS06}\label{t_CauchySubvariedad} In a globally  hyperbolic spacetime $\m$, any  compact spacelike acausal  submanifold $H$ with boundary   can be extended to a spacelike Cauchy hypersurface (thus, it  lies in a slice of some Cauchy temporal  splitting as \eqref{e_orthgonalsplitting}).  
\end{thm}
For  the proof, take a Geroch's time function $t$ in \eqref{e_Geroch}. By compactness, $H$ will be contained in the region between two $t$-slices $I^+(S_{t_-})\cap I^-(S_{t_+})$, which is globally hyperbolic. Thus, one has just to find a closed, connected, smooth, spacelike hypersurface $S$ therein containing $H$. 

It is worth pointing out  that if the acausality of $H$ were weakened into achronality, then an achronal Cauchy hypersurface $S_H$ containing $H$ will exist, but it may be non-smooth \cite[Ex. 3.4]{BS06}.  Indeed,
$\m \setminus (J^+(H)\cup J^-(H))$ is also globally hyperbolic (even if non-connected) and any of its Cauchy hypersurfaces  join together   $H$,  yielding the required
$S_H$, 
\cite[Prop. 3.6]{BS06}. 

However, the acausality of $H$ permits to find a bigger compact spacelike acausal hypersurface with boundary $G$ whose interior contains $H$. Then,  the idea is to modify the non-spacelike Cauchy hypersurface  $S_G$ maintaining $H$ unaltered.     For this purpose, one constructs a function $h$ working as in \S \ref{s_BS03} but now: (i)  $h\geq 1$ on $S_G$, $h(H)\equiv 1$ and $h>1$ on $S_G$ outside a compact subset,  (ii) the support of $h$ lies in the future of another slice $I^+(S_{t_0})$  ($t_-<t_0<t_+$), and (iii) $\nabla h$ is either  past-directed timelike or 0 on $J^-(S_G)$.  Then, the required hypersurface becomes $h^{-1}(1)$.

\subsubsection{Problem \ref{F3}: existence of a temporal function.}\label{s_F3} 
A spacetime $\m$ is called {\em stably causal} when its metric $g$ admits a 
$C^0$-neighborhood $\mathcal{U}$ (in the set of all the metrics) such that every $g'\in \mathcal{U}$ is causal. 
This is a conformally invariant property and the $C^0$ topology induces a natural one in the conformal quotient. It turns out that stable causality is equivalent to the existence of a causal Lorentzian metric $g'$ with strictly wider cones, i.e., satisfying: 
\begin{equation}\label{e_prec}
g(v,v)\leq 0 \Rightarrow g'(v,v) < 0, \; \forall v\in T\m\setminus \mathbf{0}, \qquad \hbox{denoted} \quad    
g\prec g',
\end{equation}
(see \cite[pp. 63-64]{BEE})\footnote{\label{f_Lerner} One can define directly the {\em interval topology} in the set of all the classes of Lorentzian metrics up to a conformal factor in $\m$, by declaring as open the conformal classes whose representatives lie in a set type $U_{g_1,g_2}:=\{g: g_1\prec g \prec g_2\}$, for some Lorentzian metrics $g_1,g_2$.  This topology coincides with the quotient of the 
$C^0$-topology, see \cite[p. 23]{Lerner}. This  reference is the basic one for all the properties of the $C^0$ topology to be used here; notice, however, that we will  use neither its Lemma 4.12 about stability of geodesic completeness nor its corollary, which should be compared with   \cite[\S II and Th. 4.6]{RS} and \cite{BEE} (2nd Edition, p. 240), see also footnote \ref{f_Lerner2}.}. Hawking \cite{Haw} (see also \cite[Prop. 6.4.9]{HE}) proved:
\begin{equation}\label{e_charact_stablecaus}
\exists \, \hbox{temporal func.} \; \tau \Rightarrow  \hbox{stably causal},
 \quad 
 \hbox{stably causal}  \Rightarrow  \exists \, \hbox{time func.}\; t. 
\end{equation}
He cites Seifert's \cite{Se} to justify that a  time function $t$ implies the existence of a temporal one $\tau$. Anyway, the previous techniques can be applied to complete these implications.

\begin{thm}\label{t_SaoPaulo} \cite[\S 4]{Sa05}
If a spacetime admits a time function $t$ then it admits a temporal function $\tau$. Thus, stable causality is equivalent to the existence of a time function.
\end{thm}
The proof is a simplification of the one  in \S \ref{s_CauchyTemporal_construction} for the construction of a Cauchy temporal one $\tau$. Indeed, a difficulty for the latter was to ensure that each slice $S$ of $\tau$ remained between two Cauchy hypersurfaces (indeed,  two Geroch's slices)  so that $S$ itself were Cauchy. 
Now, we will work with  Hawking's time function $t$ (whose slices migth be non-connected or have a changing topology) as with Geroch's one, but not taking care of this property\footnote{Notice, for example, that the procedure of \S \ref{s_BS03} would  also yield a ``smoothing'' of the level $t=0$ if it is applied to Hawking's time function.}, that is, we will have to ensure only that $\nabla \tau$ is timelike everywhere. 

Specifically, consider the construction of the step temporal function $\tau_0$ at the beginning of \S \ref{s_CauchyTemporal_construction}. Now, weaken the condition (iv) therein into $\nabla \tau_0$ non-vanishing only in one slice. Then, given a compact $K$, 
construct for each point $p\in K$ a function as $\tau_0$ with gradient non-vanishing in the slice $t=t(p)$.
A suitable  sum $\tau_K$ of a finite number of these functions  will satisfy that $\nabla \tau_K$ is timelike on $K$. Finally, consider an exhaustion $\{K_j\}_{j\in\N}$ of the whole  $\m$ by compact sets and the searched temporal function will be again a   the smooth sum $\tau= \sum_j \tau_{K_j}/C_j$, for a suitable choice of constants $C_j's$. 


\begin{rem}\label{r_narices} For future referencing, let us check that the existence of a temporal function $\tau$ not only implies  that $\m$ is stable causal, but also that {\em $\tau$ remains stable as a temporal function
for all the metrics in a $C^0$-neighborhood of the original one $g$} (compare also with \cite[Prop. 6.4.9]{HE}). 

Let $T_0=-\nabla \tau/|\nabla \tau|$ and consider the decomposition $T\m = $ Span$(T_0)\oplus$ Span$(T_0)^\perp$. Notice that the bundle Span$(T_0)^\perp$ is identifiable to the tangent space of the distribution provided by the slices of $\tau$ and the restriction of the metric $g$ provides a 
Riemannian metric $g^R$ on the bundle\footnote{One can write $g^R$ as $g_\tau$ in \eqref{e_orthgonalsplitting}, but we prefer to maintain $g^R$ here because there is no a splitting of $\m$ and $g^R$ is the proper geometric object with no identifications.} Span$(T_0)^\perp$.  Any $v\in T\m$ admits an orthogonal decomposition $v=aT_0 +w$ and, then, $g(v,v)=-a^2 +g^R(w,w)$. 

Now, replace $g^R$ by any 
Riemannian metric $\bar g^R$ such that $\bar g^R <g^R$ (for example $\bar g^R=g_R/2$) and consider the new Lorentzian metric $\bar g$. 
The required $C^0$ neighborhood is $\mathcal{U}=\{g' \; \hbox{Lorentz metric on $\m$} : g'\prec \bar g\}$.
Indeed, by construction $g\prec \bar g$, i.e., $\bar g\in \mathcal{U}$. Moreover, $\bar g$  (and, thus, any $g'\in \mathcal{U}$) still admits  $\tau$ as a temporal function because the tangent space of the $\tau$-levels lie outside the causal cone of $\bar g$, thus, they  are spacelike hypersurfaces
and the $\bar g$-gradient of $\tau$ must be timelike.
\end{rem}  

\begin{rem}
It is worth emphasizing that later Minguzzi  \cite{Mi10} will give  a direct proof of the fact that  time functions imply stable causality. Such a proof is a consequence of a striking  development which also links  causality   and  utility theory in Economy and provides the equivalence between stable causality and the so-called $K$-causality.
\end{rem}

\subsubsection{Problem \ref{F4}: Isometric embeddability}\label{s_F4} The isometric embeddability of globally hyperbolic spacetimes in some $\Lo^N$
 was claimed by Clarke in 1970 \cite{Cl}. However, his proof used a causally constructed embedding, non-necessarily smooth \cite[Appendix]{MS} and the problem remained dormant\footnote{After this first article, the author took into account carefully  issues about topological and differentiable structures; for example, in the article \cite{ClarkeNewman} (cited in \S \ref{s_BS03})  Geroch's splitting  is used only at the topological level.}. Anyway, a completely different proof was carried out in 2011  by using a different approach based in previous techniques. 
 Indeed, the isometric embeddability of spacetimes was completely characterized   by introducing the notion of {\em steepness} as follows.
 
 \begin{thm} \label{t_isomNash}\cite{MS}
(1) A $C^3$ spacetime $\m$ is isometrically embeddable in $\Lo^N$ if and only if it admits a temporal function $\tau$ which is {\em steep}, that is, satisfying $|\g(\nabla \tau, \nabla \tau)|\geq 1$. In particular, non-stably causal spacetimes are never isometrically embeddable.
 
 (2) Assume that $\m$ is stably causal but not globally hyperbolic. Then, its conformal class admits  representatives both isometrically embeddable and non-embeddable in $\Lo^N$. 
 
(3) A globally hyperbolic spacetime $\m$ admits a steep Cauchy temporal function. In particular, if it  is $C^3$ then it is smoothly isometrically embeddable in $\Lo^N$ and it admits a Cauchy temporal splitting as in \eqref{e_orthgonalsplitting} with $\Lambda\leq 1$.
 \end{thm}
About the proof, the necessary condition in (1) appears because, when $\m\hookrightarrow \Lo^N$, the  natural coordinate $x^0=t$ of $\Lo^N$ restricts to a steep temporal function on $\m$. For 
the sufficient one, using Remark \ref{r_narices} the metric can be written on $T\m$ as $g=-\Lambda d\tau^2 + g_\tau$ with $\Lambda \leq 1$. Then, the Riemannian metric $g_R=(4-\Lambda) d\tau^2 + g_\tau$ 
admits a Nash embedding\footnote{A different question is to find the minimum $N_0$ so that the Riemannian embedding (and, then, the corresponding Lorentzian one) exists, which may have interest in branes and higher dimensional theories. There is an  extensive mathematical theory, including \cite{Gromov} or the book \cite{HH}.} $i_0: (\m, g_R)\hookrightarrow \R^{N_0}$, and the required Lorentzian embedding $i: \m \hookrightarrow \Lo^{N_0+1} $ can be chosen as $i(p)=(2\tau(p), i_0(p))$, \cite[Prop. 3.4]{MS}. 

About (2), notice first that, when  $\m\hookrightarrow \Lo^N$,   the Lorentzian distance $d$ in $\m$ is bounded by the one in    $\Lo^N$ (i.e., $d(p,q)\leq d_{\Lo^N}(p,q)$ for all $p,q\in\m$) and, thus, $d$ is finite. Then,  a non-embeddable representative appears because the compactness of every $J(p,q)$ is equivalent to the finiteness of the Lorentzian distance function for all the conformal class \cite[Th. 4.30]{BEE}. 
Moreover, if $\tau$ is any temporal function for $g$, an embeddable representative is just   $g^*= |\nabla \tau|^2g$, as the $g^*$- gradient of $\tau$ is  unit.

Taking into account the point (1) and the expression \eqref{e_orthgonalsplitting}, the point (3) is  reduced to obtain a Cauchy temporal function with the additional restriction of steepness (a question which  may also have interest for other purposes). This was carried out in \cite[\S 4]{MS} by using techniques in the previously stuided  problems, but the following  noticeable difference must be borne in mind. As explained below Th. \ref{t_SaoPaulo}, the proof of the existence of a time function in stably causal spacetimes is  a simplification of the technique for the existence of a Cauchy temporal function. However, as we have just seen, not all the stably causal spacetimes  admit a steep temporal function and, thus, the present technique for the globally hyperbolic case {\em cannot} work in the stably causal one. 

The new genuinely globally hyperbolic ingredient in the process is the use of compact sets type $J(p,S)=J^+(p)\cap J^-(S)$, where $S$ is a Cauchy hypersurface. For any neighborhood  $V\supset J(p,S)$ one can construct a function $h_{p,S}$  ({\em steep forward cone}) with 
(i)~support in $V$  
(ii) $h_{p,S}\geq 1$ in $J^+(p)\cap S$, 
(iii)~$\nabla h_{p,S}$ either past-directed timelike of 0 in $J^-(S)$ and (iv) steep in $J(p, S)$. Such a $h_{p,S}$ is constructed by dividing $J(p,S)$ into a finite number of thin   strips, adjusting  semilocal  temporal functions for each strip
and extending $h_{p,S}$ beyond $S$ allowing  the loss of the timelike character of its gradient, \cite[Prop. 4.2]{MS}.

Now, given two slices $S_0, S_1$ of Geroch's function one 
focuses on the strip $J(S_0,S_1)$ and find a function $h^+_0$ with (i) support  in  $J(S_{-1},S_2)$, (ii)  $h_0^+ > 1$ in 
$S_1$, (iii) $\nabla h^+_0$ is either past-directed timelike of 0 in $J^-(S_1)$ and (iv)~$h_0^+$ is steep in $J(S_0,S_1)$. Such a 
$h^+_0$  with support in $J(S_{-1},S_2)$ is obtained by means of a locally finite sum 
of functions $h_{p_j,S_1}$, \cite[Lemma 4.7]{MS}.
Then,  proceed inductively by constructing a function $h^+_1$ with support in $J(S_1,S_3)$ and analogous properties with two cautions: first, the 
hypothesis (ii) now becomes $h_1^+>2$
and, second, {\em impose 
additionally that $h^+_0+h^+_1$ remains steep in $J(S_1,S_2)$ } (recall that  $\nabla h^+_0$ was not always timelike there). 
Finally, the required $\tau$ is the sum  $\sum_{m=0}^\infty \nabla h^+_m$ (minus an analogous function constructed towards the past).

\begin{rem} \label{r_steep_noparastablycaus}  
The full procedure provides a simpler proof of the existence of 
a Cauchy temporal function than the one in \S \ref{s_CauchyTemporal_construction}. However, as explained above, it is  not  
applicable (at least directly) to the existence of temporal functions.
\end{rem}

\begin{rem} \label{r_olaf}
M\"uller  emphasized  the flexibility of the Cauchy temporal splittings \eqref{e_orthgonalsplitting} proving that several additional conditions can be imposed on $\tau$. First, 
if two globally hyperbolic  metrics $g_1,g_2$ share $t$ as a Cauchy temporal function, then there exists a third metric $g$ which interpolates them, in the sense that  $g=g_1$ for $t<-1$ and $g=g_2$ for $t>1$  
\cite{Mu12}. Moreover, when the globally hyperbolic metric is invariant by a compact isometry group, then the Cauchy temporal function $\tau$ can be also constructed invariant by this group \cite{Mu16}; some applications can be found in \cite{BarS19} and \cite{AFS} (the latter to be discussed in the sketch of the proof of Th. \ref{t_princ_borde}). 
\end{rem}
 
 \subsubsection{The analytic case}\label{ss_analytic}
The  previous results admit natural extensions to the analytic case. The key is to apply  
   Grauert's criterion in \cite[Prop. 8]{Grauert} (see also Whitney's in Remark \ref{r0}):  any $C^k$ 
  function, $k\in \N$, on a  real analytic manifold can be $C^k$ approximated by analytic functions. 
  
  \begin{thm}\label{t_analyticTemporal}
  Any analytic globally hyperbolic spacetime admits an analytic steep Cauchy temporal function $\tau_{\hbox{\tiny{an}}}$.
\end{thm}    
To prove this, choose a steep Cauchy temporal function   $
  \tau$ such that $g(\nabla \tau, \nabla \tau)<-2$. Then, it is easy to find a $C^1$ neighborhood $\mathcal{U}$ of $\tau$ so that any $C^1$-
  function $t\in \mathcal{U}$ is Cauchy temporal (take $\mathcal{U}$ so that $g(\nabla t, \nabla t)<-1$ and  $|\tau-t|<1$, so that each slice $t=c\in\R$ must lie between the slices $\tau=c-1, \tau=c+1$ and, so, it will be Cauchy) and apply Grauer's.
  
\begin{rem} Such an analytic function $\tau_{\hbox{\tiny{an}}}$ can be used to construct an isometric embedding in some $\Lo^N$ by reasoning as in Theorem \ref{t_isomNash}. The only differerence it that  one must  claim Grauer's analytic isometric theorem  \cite[\S 3]{Grauert} now, instead of  Nash's.
  \end{rem}
  
In general an analytic spacelike compact acausal submanifold with boundary $H$ embedded in an analytic globally hyperbolic spacetime cannot  be extended to an analytic hypersurface.   
 \begin{example} By using Grauert's criterion, it is easy to ensure the existence of an analytic globally hyperbolic metric $g$ in $\R^2$ such that  
 $$g\preceq g_{\hbox{\tiny{1}}}:= -dt^2+dx^2, \quad \forall x\leq 1, \qquad \hbox{and} \qquad  g_{\hbox{\tiny{2}}}:= -dt^2+dx^2/9 \preceq g, \quad 2\leq x.$$ 
 Then, the submanifold $H:=\{(x,2x): 0\leq 1\}$ is analytic spacelike compact and acausal, but its unique analytic extension cannot be spacelike in $x>2$.
 \end{example}
This example shows that Theorem \ref{t_CauchySubvariedad} cannot be extended to the analytic case\footnote{Of course, applying Grauert's criterion, one could still find a spacelike Cauchy hypersurface $S$ extending any submanifold arbitrarily close to $H$.}; the possibility to extend  Theorem \ref{t_LMP} might deserve further attention.

\subsection{Further developments: Dynamical systems and cones} \label{ss_dynamical}
Next, let us discuss further results
 obtained in the framework of dynamical systems  
 for more general cones, \S \ref{s241}, \S \ref{s242}. It is worth pointing out that the approach developed in  \S \ref{s_folk} is also   applicable to some cone structures, but 
 the comparison with the previous techniques require some elements of Finsler Geometry. 
These settings lie outside our scope; anyway, we give a brief outline (adding some extra information about Finslerian elements  in footnotes). 
With this aim, we will consider first in \S \ref{s243} regular cone structures,  
 as those studied systematically in \cite{JS20}. In this case, one can find a compatible Lorentz-Finsler metric and, then, to apply the techniques in \S \ref{s_folk}. After this, we will make some comments on more general cone structures \S \ref{s245}.

 \subsubsection{Fathi and Siconolfi's cone structures}\label{s241} In 2012,  the case of a more general  cone structure $\C$,  
which generalizes the future-directed causal one of a spacetime, was considered by Fathi and Siconolfi \cite{FS}. Indeed, they assume that the cones at each point are just
    convex, closed,  not containing any complete affine line, with non-empty
interior and varying continuously. Classical notions such as {\em causal curve} (now implicitly future-directed), stable causality or global hyperbolicity can be transplantated directly to this setting. However, their approach was very different to the previous one, as it comes from weak KAM theory (which provides a  qualitative analysis of Hamilton-Jacobi equations). 
Summing up, they proved the following.
 
 \begin{thm}\cite{FS} Let $\C$ be a cone structure as above.
 
 (1) If $\C$ is  stably causal then it admits a smooth time function (strictly increasing on causal curves).
 
 (2) If $\C$ is globally hyperbolic then it is stably causal and it admits a smooth Cauchy  time function.
\end{thm}  

\begin{rem} As Bernard and Suhr emphasized in \cite{BS},  
  Sullivan   had already introduced a related approach in 1976  (see  \cite{Sullivan}), studying cone structures in  the setting of dynamical systems. At the end of this reference,   the author added a comment pointing out that his results implied  the existence of  smooth time functions on compact subsets remarking:    
``Hawking's theorem suggests the basic technique here can be extended to noncompact
manifolds'' \cite[p. 249]{Sullivan}. 

In 2014,
Monclair thesis  \cite{Monclair1} (see also \cite{Monclair2}) introduced a remarkable link between attractors, chain recurrent points and the existence
of time functions in Lorentzian manifolds by using Conley theory. Indeed, he
obtained a new original proof of Hawking's result about the existence of time functions in stably causal spacetimes
by exploiting their similarities with Lyapunov functions.  
\end{rem}

 \subsubsection{Closed cone structures}\label{s242}

Bernard and Suhr \cite{BS} developed further these ideas  and proved the existence of a temporal function for closed cone structures. 
 These structures are more general than Fathi and Siconolfi's,  because the set of all the cones is assumed to be closed in $TM$ but its pointwise variation becomes just semi-continuous (moreover, other weakenings in the hypotheses are permitted, such as the collapse of the cone to a half line or the existence of points with no cone). It is worth pointing out that, for these structures, the notion of global hyperbolicity must be strenghtened by requiring that $J(K,K')$ is compact for any compact sets $K,K'$ (as a difference with the spacetime  case,  this is not deduced now from the compactness of the diamonds $J(p,q)$). Further   issues where studied by them in a  subsequent article, \cite{BS2}.
 


\subsubsection{Regular cone structures} \label{s243} 
Following \cite{JS20}, now we consider a cone structure  $\C$  as in Fathi and Siconolfi's but assuming additionally: 
  (i)~the
  cone at each point $\C_p \subset T_p\m\setminus \{0\}$ is smooth and strongly convex\footnote{This is a standard hypothesis which can be characterized in several ways \cite[\S 2]{JS20}. For example, the intersection of $\C_p$ with any transversal hyperplane has positive definite second fundamental form (with respect to any Euclidean scalar product in $T_p\m$).}  and 
  (ii) it varies smoohtly with the point. 
 
Such cones 
 can be described nicely as the lightlike vectors of some auxiliary  Lorentz-Finsler $G$, which will be called {\em compatible} with $\C$. Simply, take any  timelike vector field $T$ of $\C$ and any 1-form $\Omega$ such that $\Omega(T)\equiv 1$, then $\C$ determines  a suitable Finsler metric $F$ in the kernel bundle  ker($\Omega$), 
so that 
a required compatible Lorentz-Finsler metric becomes\footnote{Such metrics are positive 2-homogeneous, i.e,  $G(a v)=a^2G(v)$ for $a>0$. A small subtlety it that the metric $G$ above is not smooth in the direction of $T$ (this is a general issue in the Finslerian setting, coming from the non-smoothness at 0 of the square of any norm which does not come from a scalar product). However, it will not have any relevance here, as  only  the directions around the lightcone will play a role. Anyway, $G$ can be smoothed around $T$ preserving $\C$, see \cite[Th. 5.6]{JS20}.} 
 $G=\Omega^2-F^2$ \cite[Th. 1.2]{JS20}. In the particular case of a Lorentzian metric $g$, the vector field $T (\equiv T^a)$ corresponds with any unit (future-directed) timelike one,  $\Omega$ can been chosen as its $g$-dual ($\Omega_b\equiv T_b$) and $F$ will be then the 
 norm associated with the Riemannian metric $g^R$ induced by $g$ in $T^\perp$, so that  
 $$G(v)=\Omega(v)^2-F^2(v)=g(T,v)^2-g^R(v^\perp,v^\perp)=-g(v,v), 
 $$ 
where $v^\perp=v-g(v,T)T/g(T,T)$ 
(the domain of $G$ will be assumed only the causal cone so that\footnote{The Lorentz-Finsler choice of signature is $(+,-,\dots,-)$, which applies to the  Hessian of $G$ at each $p\in \m$.}  $G(v)\geq 0$. 

This provides a picture  similar to the one in Remark \ref{r_narices}. Notice, however, that the distribution $T^\perp$ may be non-integrable now. 
In particular, if $F$ is replaced  by a smaller Finsler metric 
$\bar F$ (i.e.,  $\bar F<F$ on $T\m\setminus \mathbf{0}$), the corresponding metric $\bar G$ will have  a wider cone structure $\bar \C$ (i.e., $G\prec \bar G$). 
  
\subsubsection{Applicability of the Lorentz setting for regular cones}  \label{s244}
  For any regular cone structure, the existence of the Lorentz-Finsler metric  $G$ allows one to extend classical causality  (including Geroch's theorem) to $\C$. 
 We will say that a function $h$ is temporal for $\C$ (and, then, for $G$) when $dh(v)>0$ for any causal $v$. Analogously, a smooth hypersurface $S$ is spacelike for $\C$ (and, then, for $G$) when  its tangent space does not contain causal vectors.
  The existence of $G$ permits to define
   an analogous to the modified Lorentz distance functions  in \S \ref{s_BS03}, and one can check that
     {\em all the constructions of spacelike Cauchy hypersurfaces and temporal functions in \S \ref{s_folk} can be reproduced directly for a regular cone structure $\C$,} 
following the ideas  sketched at each case.
In particular, if $\m$ admits a globally hyperbolic regular cone structure it will decompose smoothly as a product $\R\times S$ with Cauchy slices.
 
 However, the following subtleties deserve to be stressed for the other issues studied in \S \ref{s_folk}. 
When $h$ is temporal function, in the  Lorentzian case, the direction of its gradient 
depends 
 only on $\C$ (as two Lorentzian metrics sharing the cones are conformal), but in the Lorentz-Finsler case it depends also on the chosen compatible\footnote{The gradient is now related to the Legendre transform and admits the following geometrical interpretation. The function $\tau $ is temporal   
when, at each $p\in\m$, the cone $\C_p$ is intersected transversaly by some affine hyperplane $v_0+$  ker$(dh)_p$ such that $dh_p(v_0)>0$. Now, if $G$ is a compatible  Lorentz-Finsler metric, there will exists some $a>0$ 
such that $av_0+$  ker$(dh_p)$ is tangent to  the {\em indicatrix} $\Sigma_p$   of $G$ at $p$ ($\Sigma_p$ is defined as the set of the $G$-unit vectors; it is concave and asymptotic to the cone). This selects the point of tangency $u\in \Sigma_p$, which will provide the direction of the $G$-gradient $(\nabla^G h)_p$; indeed, $(\nabla^G h)_p= ((dh)_p(u))u$.}  $G$. Now notice: 
 
\ben\item  Once a Cauchy temporal function $\tau$ has been obtained for $\C$ (say, by using some auxiliary Lorentz-Finsler $G$), one would  consider $\nabla^G \tau$ in order to construct a smooth Cauchy splitting with spacelike Cauchy slices as in \S \ref{s_CauchyTemporal_construction}. 
However,  {\em an analogous expression to  the orthogonal splitting
  \eqref{e_orthgonalsplitting} would not be found  
  for\footnote{Recall that, as a difference with Lorentz metrics,  a Lorentz-Finsler ones cannot be decomposed in the tangent bundle as a difference $dt^2-F_0^2$, being $F$ a Finsler metric (a  subtler decomposition  as a difference between a Riemannian and a Finslerian metric is possible, see \cite[\S 4.4]{JS20}). So, even though $\nabla^G\tau$ will be orthogonal to
$TS\subset T(\R\times S)$  the expression of the metric is more complicated (indeed, it is similar  to the one explained for static spacetimes    
in \cite[\S 4.2.1]{JS20}, taking into account that, in our case, the metric is $\tau$-dependent).
} $G$}.

Anyway, once $\tau$ and $\nabla^G \tau$ are given, one  can put $T=\Lambda \nabla^G \tau$ with $\Lambda =1 /d\tau(\nabla^G \tau)$ so 
that $d\tau(T)\equiv 1$ and  determine   a second compatible 
Lorentz-Finsler metric\footnote{\label{footAnisConf} The Lorentz-Finsler metrics compatible with the same cone structure are called {\em anisotropically conformal} and they share the same lightlike pregeodesics, called then {\em cone geodesics}, \cite[Th. 1.1.]{JS20}.}  $\tilde G$  such that:
\begin{equation}\label{e_orhogonal_Finsler}
\tilde G=\Lambda d\tau^2-F^2_\tau
\end{equation}    
  (where $F_\tau$ is now a natural notation for the Finsler metric in  ker($d\tau$)) with 
  $\nabla^G \tau=\nabla^{\tilde G} \tau$. That is, the analogous to the Cauchy temporal splitting \eqref{e_orthgonalsplitting} is obtained only for some representatives of the anisotropically conformal class (according to the nomenclature in footnote \ref{footAnisConf}).

\item  When $\C$ admits a temporal function $\tau$, the reasoning in Remark~ \ref{r_narices} can be reproduced with $\nabla^G \tau$, for any $\C$-compatible $G$.  This shows that $\tau$ remains temporal for all the  cone structures in a $C^0$-neighborhood of $\C$ and, in particular, that $\C$ is stably causal, extending Th. \ref{t_SaoPaulo}.

 \item The existence of a steep temporal function for any Lorentz-Finsler metric  $G$   compatible  with   a globally hyperbolic cone structure $\C$ would be ensured. However, its relation with  Nash-type embeddings (as in Theorem \ref{t_isomNash}) is unclear. 
 
 Indeed, taking into account an expression such as \eqref{e_orhogonal_Finsler}, one should consider first the Finsler case. However,  the  results in this case  are not so tidy and this  suggests additional  difficulties for the   Lorentz-Finsler case\footnote{The analogous of a Finslerian Nash result would be to find an isometric embedding in 
 $(\R^{N_0}, \parallel \cdot \parallel)$, where  $\parallel \cdot \parallel$ is a Minkowski norm (thus, with strongly convex  indicatrix, i.e. unit sphere). Burago and Ivanov  \cite{BI} proved its existence in the compact case, but Shen \cite{ShenIsom} gave a counterexample in the non-compact one. Indeed, he proved the impossibility to find such an embedding in the case of an unbounded Cartan tensor. However, Shen cited Gu \cite{Gu}  suggesting the possibility of an isometric embedding when the indicatrix of $\parallel \cdot \parallel$ is permitted to be only convex. Thus, for a Lorentz-Finsler metric one would wonder on the possibility to embed isometrically in $(\R\times \R^{N_0}, dt^2-\parallel \cdot \parallel^2)$ for a (possibly only convex) norm $\parallel \cdot \parallel$.}. Notice, however, that if one looks for an embedding of $\C$ into the cone structure of some $(\R\times \R^{N_0}, dt^2-\parallel \cdot \parallel^2)$, it would be enough to embed isometrically one of the $\C$-compatible $G$; this might pose new Finslerian issues\footnote{For example,   whether $\C$ could be expressed as in \eqref{e_orhogonal_Finsler} by using Finslerian metrics with bounded Cartan tensor (in order to avoid Shen's result in the previous footnote).}.    
\een
 
\subsubsection{Considering more general cones}\label{s245} In order to extend the previous results  from regular cones to more general ones, notice that the conditions (i), (ii) (imposed above as additional hypotheses to Fathi and Siconolfi 
cones) are not essential for our approach. Indeed, the weakening of the strong convexity of $\C$ into convexity only affects to the uniqueness of the cone geodesics (notice \cite[Rem. 2.10]{JS20}), but this does not play any role here. Moreover, the lack of smoothness of both each cone $\C_p$ (including its collapse  to a half line) and its variation with $p$, can be overcome by taking into account that the modified Lorentzian distances in \S \ref{s_BS03} are constructed locally and, then, one could choose wider regular cones. 
Indeed, even the non-continuous variation of the cones would be permitted whenever Geroch or  Hawking method to construct a time function worked. In these general cases, $\C$ may not have a compatible $G$; thus, one would obtain a temporal function but not a meaningful gradient   (even though it is not difficult to figure out particular situations where $\nabla^G\tau$ would make sense for more general versions of $G$).  
 
These links between dynamical systems and cones become exciting, suggest the possibility to enrich both fields and should attract further attention.





 \subsection{Revisiting Geroch's stability and Seifert's smoothability}\label{s_Festability} 
Among the previous  issues on stability, we mentioned in \S \ref{s_BS03} that  Geroch gave an argument to justify the $C^0$-stability of global hyperbolicity. In 2011, Benavides and Minguzzi  \cite{BM} critiqued this Geroch's proof  emphasizing that it would not work without introducing non-trivial amendments such as the Cauchy temporal function in\footnote{See \S 3 in the arxiv version of \cite{BM}.} \cite{BS05}.
Then, they gave  a direct and strictly topological proof of the stability of global hyperbolicity which does not use the concept of Cauchy hypersurface or the topological splitting. As a corollary, they obtained that every globally hyperbolic spacetime admits a Cauchy hypersurface that remains Cauchy for small perturbations of the spacetime metric. The conclusions of this result can be improved because  any Cauchy temporal function (and, thus, any spacelike Cauchy hypersurface, Th. \ref{t_LMP}), remains Cauchy temporal for all the metrics in a $C^0$ neighborhood\footnote{See M. S\'anchez, {\em A note on stability and Cauchy time functions},     arxiv: 1304.5797, which was included later in the more general setting of \cite{AFS}. 
} in the same vein as explained for temporal functions in  
Remark~\ref{r_narices}; the key is that these properties are locally stable and the $C^0$ topology is so fine that will ensure global stability. Anyway, their viewpoint was interesting and  fruitful for the following developments.

Indeed, after this article, Chrusciel, Grant and Minguzzi \cite{CGM} 
revised Seifert's technique, developed the notion of anti-Lipschitz function and carried out a proof of smoothability of time and Cauchy time functions following Seifert's viewpoint. 
Minguzzi also reobtained the existence of steep Cauchy temporal functions \cite{Mi16}, the stability of Cauchy temporal functions \cite[Thm., 2.46]{Mi17} and  the stability of spacelike Cauchy hypersurfaces (and, in general, of locally stable Cauchy hypersurfaces) \cite{Mi19} by following this approach. What is more, in the last two references   he develops the approach in the aforementioned setting of closed cone structures.
So, this approach 
joins to the previously mentioned ones  broadening the scope and applicability of techniques of classical general relativity in  
Finsler and non-regular spacetime geometry.

\section{Global hyperbolicity 
for sliced spacetimes}\label{s4}

Taking into account the Cauchy temporal splitting \eqref{e_orthgonalsplitting}, next we will wonder about a converse, namely when a spacetime which admits a split  expression $\m= \R\times M$ is globally hyperbolic, being its $\R$-projection a Cauchy temporal function. This will lead  to some sufficient conditions, which will have interest for several purposes.

\subsection{Slicings by parametrized products}\label{s_slicing_param_product}
Consider the normalized sliced spacetime $\m$:
\begin{equation}\label{e_orthogonalnormaliz_spl}
 I \times M, \qquad g=-d\ttt^2 +g_\ttt, \qquad I=(a,b)\subset \R \; \hbox{(interval)}
 \end{equation}
where  $g_\ttt$ is again a $\ttt$-parametrized Riemannian metric on each slice. In comparison with the Cauchy temporal splitting \eqref{e_orthgonalsplitting}, we have set $\Lambda\equiv 1$ because now we are interested in conformally invariant properties such as  global hyperbolicity (i.e.,   we would choose the conformal representative $g/\Lambda$  in \eqref{e_orthgonalsplitting}). Notice also that the case $I \subsetneq \R$ can be always reduced to the case\footnote{On the one hand, taking any increasing diffeomorphism $\phi: I\rightarrow \R$ the new temporal function $\tilde\ttt=\phi\circ\ttt$ yields a splitting with $\R\times M$ where the original slices are only relabelled. On the other, any strip $I\times M\subset \R\times M$ of a Cauchy temporal splitting \eqref{e_orthgonalsplitting} is intrinsically globally hyperbolic and the slices remain Cauchy.} $I=\R$ but the interval $I$ will be convenient for some examples.
Trivially, $\ttt$ is a temporal function in \eqref{e_orthogonalnormaliz_spl} and we will wonder when it is Cauchy.

\subsection{The role of the completeness of the slices} The completeness of the slices of a  globally hyperbolic spacetime may be important  for different purposes (see for example \cite{Finster}). However, this is not an invariantly conformal property. Indeed,  for Riemannian manifolds, 
 completeness can be acquired (or lost in the  non-compact case) by means of a conformal change, see \cite{NomOz}. The proof of this fact   applies  to the slices of any Cauchy temporal splitting  $\R\times M$ 
showing: {\em if the manifold $M$ is non-compact, then there exists  a conformal factor such that all the Cauchy slices are complete as well as another factor such that all of them are incomplete.}

However, the situation is  different in the case of a metric as in \eqref{e_orthogonalnormaliz_spl}. Recall that the lapse $\Lambda$ in the case of a steep Cauchy temporal function satisfies $\Lambda\leq 1$, thus, in this case the normalized conformal representative $g/\Lambda$ will have a  Riemannian part satisfying $g_{\ttt}/\Lambda\geq g_{\ttt}$ for every slice, that is, not only  the metrics $g_\ttt/\Lambda$ will be complete whenever so are the $g_\ttt$'s, but also they could be complete even if the $g_\ttt$'s are incomplete. However, the example in \S \ref{ss0} shows: {\em even if $\ttt$ is Cauchy temporal for the normalized splitting \eqref{e_orthogonalnormaliz_spl}, incomplete slices may appear}. 

In spite of  examples as the previous  one (which is rather artificial), the hypothesis that all the slices $g_\ttt$ are complete is natural to ensure that $\ttt$ is Cauchy for the sliced spacetime \eqref{e_orthogonalnormaliz_spl}. However, this condition is {\em not sufficient, as shown by the non-globally hyperbolic example in \S \ref{ss1}}.

\subsection{Uniform bounds for the slices}\label{s_unifbound} A simple sufficient strengthening of the completeness of all the $g_\ttt$'s is the following.
\begin{prop}\label{p_Bari}
Consider the sliced spacetime $\m=(I \times M,g)$ in \eqref{e_orthogonalnormaliz_spl}. The temporal function $\ttt$ is Cauchy (and, thus, $\m$ globally hyperbolic) if there exists a complete Riemannian metric $g_R$ with associated distance $d_R$ on $M$, a point $\hat x\in M$ and a  continuous function $m: I\rightarrow (0,\infty)$ such that:
\begin{equation}\label{e_unif_bound_gh}
g_{\ttt_0}(v,v) \geq \frac{m(\ttt_0)}{1+d_R(\hat x,x)^2} \, g_R(v,v), \qquad \forall v\in T_xM, \forall x\in M, \forall \ttt_0\in I.
\end{equation}
\end{prop}
Obviously the best choices for $m$   make $\lim_{t\rightarrow a}m(t)=\lim_{t\rightarrow b}m(t)=0$.

This result is a particular case of those in \cite{S}. Anyway,  it is trivial if $g_{\ttt_0}(v,v) \geq g_R(v,v)$,
because $\ttt$ is Cauchy temporal for the static metric $-d\ttt^2+g_R$ and, under this assumption, this metric have wider cones than $g$. 
The conformal factor $1/(1+d_R(\hat x,x)^2)$ does not affect to the completeness of the associated distance and, thus, neither to the Cauchy character of the slices. Moreover the  function $m(\ttt) >0$ can also be incorporated as a factor because, if there were, say,  a future-directed, inextendible  $\ttt$-parametrized curve $[\ttt_0,\ttt_1)\ni \ttt\mapsto (\ttt,c(\ttt))$ with $t_1<\infty$, one can take the minimum $m_{-}$ of $m$ in  $[\ttt_0,\ttt_1]$, replace $g_R$ by $m_{-} \; g_R$ and obtain a contradiction. 

\begin{rem}\label{r_linksLF0}
The general sufficient conditions  in \cite{S} include the case when there are cross terms type $d\ttt\otimes \omega_\ttt+ \omega_\ttt\otimes d\ttt $ for some $\ttt$ dependent 1-form. 

In the $\ttt$-independent case (i.e., the spacetime is stationary with $\omega_\ttt\equiv \omega_0$ and $g_\ttt\equiv g_0$ )
the fact that the slices are Cauchy is characterized by the completeness of the {\em Fermat metric} $F=\sqrt{g_0+ \omega_0^2}+\omega_0$, which is a Finsler metric (of Randers type), see \cite[Th. 4.4]{CJSa}. Notice that, essentially, a Finsler metric replaces  the pointwise scalar products of Riemannian Geometry by  norms. However, these are allowed to be non-reversible (i.e., $F(-v)\neq F(v)$ in general), and this property becomes essential for the characterization of global hyperbolicity \cite[Th. 4.3]{CJSa} (there are globally hyperbolic examples with non-Cauchy slices).

In the $t$-dependent case, a  characterization with a parametrized $F_\ttt$ similar to  
\eqref{e_unif_bound_gh} works, the details can be seen in \cite{Tapia}.
\end{rem}
\begin{rem}\label{r_ej3}
One could wonder whether uniform bounds such as  
\eqref{e_unif_bound_gh} are too restrictive. Notice, however, 
that the example in \S \ref{ss2} lowers expectations. Indeed, it shows that 
if $g$ admits a normalized Cauchy temporal splitting as \eqref{e_orthogonalnormaliz_spl},
{\em a change of the 
parametrized metrics $g_\ttt$  by just   $g_\ttt /2$ may yield a non-globally hyperbolic 
spacetime} (compare with Rem. \ref{r_narices}).
\end{rem}

\subsection{Non-convexity of  global hyperbolicity in  conformal classes}\label{s_paracausal}
Taking into account the procedure to construct the  splitting \eqref{e_orthgonalsplitting} from a Cauchy temporal function
(see above that formula), we can assert that any globally hyperbolic metric on $\R\times M$ admit a  conformal diffeomorphism with one of the metrics written as  in the normalized splitting  \eqref{e_orthogonalnormaliz_spl} (with $I=\R$). All these representatives of the conformal classes\footnote{The representative of a single conformal class will be non unique (indeed,  different choices of the Cauchy temporal function will give different representatives) but this will not be relevant for our purposes, in general.} will share $t$ as the same Cauchy temporal function and will differ only in the $g_t$ part. This simplified setting is natural to pose some questions regarding the space of globally hyperbolic metrics on $\R\times M$.

\subsubsection{Existence of non-globally hyperbolic convex combinations}
Notice that the space of normalized sliced metrics \eqref{e_orthogonalnormaliz_spl} on $\m$ is convex, that is, given two such metrics $g, g'$, all the {\em convex combinations} $\lambda g+(1-\lambda)g'$ with $\lambda\in [0,1]$ are also  metrics of this type\footnote{This relies on the convexity of the set of all the Riemannian metrics, as $\lambda (-dt^2+g_t)+(1-\lambda)(-dt^2+g'_t)=-dt^2+ (\lambda g_t+(1-\lambda)g'_t)$.}.

Then, it is natural to wonder whether, in the case that $g$ and $g'$ share $t$ as a Cauchy temporal function, this property will hold for their convex combinations.  The answer to this question is negative, as  the example in \S\ref{ss3} (Theorem \ref{t_nonconvex}) shows two such   $g, g'$ such that  $(g+g')/2$ is not even globally hyperbolic.

\subsubsection{Connectivity by convex combinations with breaks}
There exists an extra motivation for this question coming from the Cauchy problem for normally hyperbolic operators \cite{FNW, MMV}. Some authors have studied the possible deformation of a representative $g$ of a conformal class into a representative $ g'$ of a different  class; in particular, {\em paracausal} transformations have been  introduced in \cite{MMV}  with this aim. 
Essentially, they realize that there is a third metric $\tilde g$  whose cones are not wider  than those of $g$ and $g'$, that 
is\footnote{The notation $\preceq$ means that inclusion of the cones may be non-strict, as a difference with $\prec$ in \eqref{e_prec}.}
  $\tilde g\preceq g$ and $\tilde g\preceq g'$. Notice that, trivially, the convex combinations of $\tilde g$ with $g$ 
(as well as with $g'$) share $t$ as a Cauchy temporal function, as the timecones of these  combinations are not wider than those of $g$. This allows one to path-connect $g$ with $g'$ by means of  two convex combinations passing through $\tilde g$, that is, by means of a {\em piecewise convex combination with one break}\footnote{Notice that the paths we are using are always piecewise convex combinations, so, we do not need to consider a topology in the space of metrics to define them. Anyway, to ensure the continuity of the path (a property that any  path should fulfill by definition) the $C^0$ topology would not by appropriate, but uniform convergence on compact sets would suffice.}. 
Even though such results were stated in the case of compact $M$,  we can use our previous techniques to obtain the following general one. 
\begin{thm}  \label{t_paracausal}
 The space of  normalized globally hyperbolic metrics which share $t$ as a Cauchy temporal function as  in \eqref{e_orthogonalnormaliz_spl} 
is piecewise convex connected; indeed, any two such metrics can be connected by a piecewise convex combination  with  one break. 

 However, this space is not convex in general, 
 as there exists two such metrics with a non-globally hyperbolic convex combination.  
\end{thm}

\begin{proof} 
Let $g_M$ be any complete Riemannian metric and consider the globally hyperbolic static product $\m_0= (\R\times M, -dt^2+g_M)$. For any    $g=-dt^2+g_t$ with $t$ Cauchy temporal, the convex combination 
$$\lambda (-dt^2+g_M) + (1-\lambda) g= -dt^2+ \lambda g_M + (1-\lambda) g_t \prec -dt^2+ \lambda g_M, \quad \forall \lambda\in (0,1),$$ 
maintains $t$ as Cauchy temporal (as $\lambda g_M$ is also complete). This would happen also for a second  $g'$, 
thus, connecting $g$ and $g'$ with a piecewise combination with a break at $\m_0$. 
The counterexample to convexity is given by the metrics  $g_3^{\hbox{\tiny{even}}}$ and $g_3^{\hbox{\tiny{odd}}}$ in Theorem \ref{t_nonconvex}. 
\end{proof}

\begin{rem} \label{r_paracausal} (1) The last assertion can be strenghtened, as the claimed space {\em is convex if and ony if $M$ is compact}. 
   Indeed, when  $M$ is  compact, 
   convexity is just a consequence of the fact that all the metrics as in \eqref{e_orthogonalnormaliz_spl} admit $t$ as a Cauchy temporal function. This property can be checked  directly by applying  Proposition  \ref{p_Bari}  being $g_R$ any (necessarily complete) Riemanian metric $g_R$ on $M$ and $m$ the continuous function
   $$
   m(t)=\hbox{Min} \{g_t(v,v): g_R(v,v)=1, v\in TM\}, \qquad \forall t\in I\subset \R.
   $$ 
When $M$ is not compact, 
take a complete Riemannian  metric on $M$ and choose one of it rays $[0,\infty)\ni s \mapsto c(s)$  
The required counterexample would be obtained by  constructing two Lorentzian  metrics $g, g'$ on $\R\times M$  such that $t$  becomes a normalized Cauchy temporal function for them and $g, g'$ behave as $g_3^{\hbox{\tiny{even}}}$ and $g_3^{\hbox{\tiny{odd}}}$   on $\R\times$  Im$(\gamma)|_{[1,\infty)} \cong \R \times [1,\infty)$. 

(2) The previous results are naturally extended to the case when  $\hat g_1, \hat g_2$ are globally hyperbolic 
metrics with arbitrary lapses $\Lambda_1, \Lambda_2$ for some  Cauchy temporal functions  $\tau_1, \tau_2$, resp.  As explained, they both will be isometric to a Cauchy temporal splitting $g_1, g_2$ as in \eqref{e_orthgonalsplitting} with the same function  $\tau$. Then,   the trivial fact that 
the convex combinations of each $g_i$ with the normalized one\footnote{   
i.e., $[0,1]\ni \lambda \mapsto\Omega_i(\lambda) g_i$ with $\Omega_i(\lambda)=1/\left(\lambda \Lambda_i +(1 - \lambda)\right)$ for $i=1,2$.}   
are  globally hyperbolic,
allows one to connect $g_1$ and $g_2$ by piecewise convex combinations which are also globally hyperbolic.
\end{rem}

\section{From the conformal to the causal boundary}
\label{s2a}


\subsection{Penrose 
embeddings in static models $\R\times M_k$}
\label{s2.1} The beautiful  open conformal embedding   of Lorentz-Minkowski space $\Lo^{n+1}= (\R\times \R^n,\gl)$  into Einstein static universe 
$(\R\times S^n, -dT^2+g_{S^n})$ (which, in certain sense, extends stereographic projection)
underlies the conformal boundary. In \S \ref{s2.2}, we will refer also to a similar  embedding where the sphere $S^n$ is replaced by hyperbolic space $H^n$. 
 Thus, for the sake of completeness,  next we  construct conformal open embeddings  from suitable open subsets of $\Lo^{n+1}$ into $(\R\times M_k, -dT^2+g_k)$, where $(M_k,g_k)$ is the simply connected model $n$-space of constant curvature $k\in \R$ (the familiarized reader can go directly to the next subsection). 
 Consider the functions 
$$S_k(r):=
\left\{ \begin{array}{lll}
\frac{\sin(\sqrt{k}r)}{\sqrt{k}} & & \hbox{if} \; k>0 \\

r & & \hbox{if} \; k=0 \\

\frac{\sinh(\sqrt{-k}r)}{\sqrt{-k}} & & \hbox{if} \; k<0 
\end{array}\right. \qquad  
C_k(r):=
\left\{ \begin{array}{lll}
\cos(\sqrt{k}r) & & \hbox{if} \; k>0 \\

1 & & \hbox{if} \; k=0 \\

\cosh(\sqrt{-k}r) & & \hbox{if} \; k<0, 
\end{array}\right. 
$$
 $ r\in \R$, with derivatives  $S'_k=C_k$, $C'_k=-k S_k$ and $C_k^2+k S_k^2=1$, and put 

$$T_k: \left(-\frac{\pi}{2\sqrt{k}}, \frac{\pi}{2\sqrt{k}}\right)\rightarrow \left(-\frac{1}{\sqrt{-k}}, \frac{1}{\sqrt{-k}}\right), \qquad r \mapsto T_k(r):=\frac{S_k(r)}{C_k (r)},$$
with the convention $1/\sqrt{k}=\infty$ if $k\leq 0$.  $T_k$ is a diffeomorphism with $T_k'=1+kT_k^2$ satisfying
$
T_k(r-r')=(T_k(r)-T_k(r'))$ $/(1+k \, T_k(r)\, T_k(r'))$.

Using normal spherical coordinates for $g_k$  
(see for example \cite[\S III.1]{Chavel}):
$$ g_k= dr^2 +S_k^2(r) \, \gs, \qquad r\in (0,\pi/\sqrt{k}), 
$$
Consider 
$\Lo^{n+1}$ 
in both standard spherical coordinates and the corresponding advanced and retarded ones $(u,v)$,  $u=t-r$, $v=t+r$:
\begin{equation}
\label{e_gl}
\gl = -dt^2+dr^2+r^2 \gs= -\frac{1}{2}(du dv+dv du)+ \frac{1}{4}(v-u)^2 \,\gs
\end{equation}
where  $v>u$ (the timelike axis at 0 is excluded). 
 Let $\ak$ be the inverse function of $T_k$ (thus, $\ak'(s)=(1+k s^2)^{-1}$)
and introduce the new coordinates: 
$$ \begin{array}{ll}
T= \ak v+ \ak u   & \left(dT=\frac{dv}{1+kv^2}+\frac{du}{1+ku^2}\right)\\ 
R= \ak v - \ak u   & \left(dR=\frac{dv}{1+kv^2}-\frac{du}{1+ku^2}\right)
\end{array} \; \hbox{for} \; \; -\frac{1}{\sqrt{-k}}< u<v< \frac{1}{\sqrt{-k}},$$
the latter giving a bounded domain  when $k<0$. The range of $T,R$ is then:
\begin{equation}\label{e_restr_TR}
-\frac{\pi}{\sqrt{k}} < T+R <\frac{\pi}{\sqrt{k}},
\qquad
-\frac{\pi}{\sqrt{k}} < T-R <\frac{\pi}{\sqrt{k}},
\qquad 0<R.
\end{equation}

\begin{prop}\label{p_PenroseTypeEmbed} In the coordinates defined above, 
$$
-dT^2 + dR^2 + S_k^2(R) \, \gs = \Omega^2 \gl \quad 
\hbox{where} \;  \Omega^2(u,v)={4}/{(1+kv^2)(1+k u^2)}.
$$

Thus, for any $k\in \R$, one obtains  an explicit conformal diffeomorphism between the following regions of $\Lo^{n+1}$ and of the product $(\R\times M_k, -dT^2+g_k)$: 
\begin{enumerate}
\item \label{item1} For $k>0$, from the whole $\Lo^{n+1}$ to the region of $(\R\times M_k, -dT^2+g_k)$
under the inequalities for $T\pm R$ in \eqref{e_restr_TR}.
\item For $k=0$, a 
homothety 
from $\Lo^{n+1}$ to $(\R\times M_0, -dT^2+g_0)$.
\item \label{item3} For $k<0$, from 
 the region satisfying 
 $-1/\sqrt{-k}< t-r$ and $t+r< 1/\sqrt{-k}$ 
 of $\Lo^{n+1}$ to the whole spacetime $(\R\times M_k, -dT^2+g_k)$.
\end{enumerate}
\end{prop}
\begin{proof}
From the expressions of $T,R$, the term in $du,dv$ of \eqref{e_gl} is $(-dT^2+dR^2)/\Omega^2$. The other term also matches by using $T_k(R)=(v-u)/(1+kvu)$:
$$
S_k^2(R)=\frac{1}{k}(1-C_k^2(R))=\frac{T_k^2(R)}{1+kT_k^2(R)}=\frac{(v-u)^2}{(1+kvu)^2+k(v-u)^2}= \frac{\Omega^2}{4}(v-u)^2.
$$
\end{proof}
\begin{rem}\label{r_PenroseTypeEmbedding}
(1) In the case $k<0$, let $r_k=1/\sqrt{-k}$ and consider the ball $B_k=\{r<r_k\}$ in the slice $\{t=0\}$ of $\Lo^{n+1}$. Its (future and past) domain of dependence, $D(B_k)=D^+(B_k)\cup D^-(B_k) (\subset \Lo^{n+1})$ is the domain of the conformal map  in the item \eqref{item3} above, that is, this map is a conformal diffeomorphism from $D(B_k)$ to $\R\times M_k$. 

(2) In particular, the previous diffeomorphism maps $B_k$
into 
 the hyperbolic space $M_k$ (regarded as the slice $T=0$ of $\R\times M_k$) and the pullback metric induced on $B_k$  
yields the classical conformal ball representation of the hyperbolic space. Analogously, for $k>1$ one should regard  $B_k=\R^n$ and the 
 restriction of the conformal map in the item \eqref{item1} is (the inverse of) the stereographic projection from the north pole. 
\end{rem}

\subsection{Looking for open conformal embeddings}\label{s2.2}
Penrose embedding $i: \Lo^{n+1}\hookrightarrow \R\times S^n$ in Prop. \ref{p_PenroseTypeEmbed}, \eqref{item1}, permits to regard the topological boundary  $\parti \Lo^{n+1}$ 
of 
$i(\Lo^{n+1})$ in $\R\times S^n$ as the {\em conformal boundary}  of $\Lo^{n+1}$. The notion of {\em asymptotic flatness} was then introduced by imposing the existence of a conformal embedding which resembles some properties of $i$,   \cite{Pe} (see also, for example,  \cite[Chapter 11]{Wald}). The successful interpretations of this boundary suggest  a natural generalization,  namely: 
given a (strongly causal) spacetime $\m$ try to find a suitable open conformal embedding $i: \m \hookrightarrow \hat \m$ into a bigger ({\em aphysical}) spacetime $\hat \m$ and regard the topological boundary  $  \partial \left(i(\m )\right)\subset \hat \m$ as the conformal boundary $\parti \m$ of $\m$. 


\subsubsection{Strictly causal elements} \label{ss2.1}
An observation 
about  $\parti \Lo^{n+1}$ is that it involves two distinct types of elements: 

(a) The strictly causal ones, which are conformally invariant, namely, the well-known points $i^\pm$ and null boundaries 
  $\mathcal{J}^\pm$
diffeomorphic to $  \R^+\times S^{n-1}$
(these lead to an asymptotic double  cone structure on some sphere).
 
(b) Those depending on the conformal factor $\Omega$ including its carefully chosen 
decay 
(see Prop. \ref{p_PenroseTypeEmbed}) which leads to 
the infinite spacelike point $i_0$. 

\smallskip
  
Here, we will focus on conformally invariant properties (as global hyperbolicity itself), dropping (b) so that there are no a priori restrictions for the embeddings. 

\subsubsection{Additional requirements on $i$} \label{ss2.2} The conformal  embedding $i: \m \hookrightarrow \hat \m$ should satisfy some hypotheses so that the properties of $\parti \m$ can reflect intrinsic properties of $\m$. 
A usual requirement is the compactness of $\m \cup \parti \m$. This is satisfied by the standard $\parti\Lo^{n+1}$, but it  would be lost if $i_0$ where removed (as suggested above). A natural weakening of compactness is   {\em chronological completeness}, i.e,  any inextendible  timelike curve  in $\m$ acquires two endpoints in $\parti \m$. This hypothesis turns out  essential  to relate the conformal and causal boundaries, see \cite[\S 4]{FHS11}.  

In particular, one can wonder about the requirements to ensure that $\m$ is globally hyperbolic. It is easy to check that, if $\parti \m$ is a  hypersurface of $\hat \m$ containing a {\em timelike
 point} $z$ (that is,  $\parti \m$ is smooth enough at $z$ so that
the tangent space $T_z(\parti \m)$ is well defined and it is a timelike hyperplane of $T_z\hat \m$) 
 then $\m$ cannot be globally hyperbolic. To ensure the converse, a simple choice of additional hypotheses is the following: 
 \begin{prop}\label{p_timelikepoint_conf} \cite[Corollary 4.34]{FHS11}
 Assume that $\parti \m$ is a $C^1$ hypersurface of $\hat \m$ and  chronologically complete. $\m$ is globally hyperbolic if and only if  $\parti \m$ does not contain any timelike  point.
 \end{prop}
  It is worth pointing out that the proof  in the cited reference shows first that, {\em under the chosen hypotheses, the conformal and causal boundaries are equivalent} \cite[Theorem 4.33]{FHS11}. Thus, the 
properties of $\parti \m$ become intrinsic to $\m$, as so is the causal boundary. This also poses the problem of choosing reasonable hypotheses for the conformal embeddings so that   the properties of $\parti \m$ will be independent of the chosen $i$.

\subsubsection{Inexistence of conformal embeddings}\label{ss2.3}
The main difficulty of the conformal boundary is that there is no a general way to find an open conformal embedding satisfying some natural hypotheses as above (or even non-being trivial). 
The following example may be illustrative, as well as a motivation for the causal boundary later.

\begin{example}\label{ex_NO_confbound}
Conformal boundary for some products $(\R \times M, g=-dt^2+\g$). When $(M,g)$ is the 
hyperbolic space $H^n$,  we can consider
the inverse $i: \R\times H^n \hookrightarrow D(B_k) \subset \Lo^{n+1}$ of the conformal embedding given in 
Prop.~\ref{p_PenroseTypeEmbed}~ \eqref{item3} (see also Rem.~\ref{r_PenroseTypeEmbedding}(1)), and 
$\parti (\R\times S^n)$ would be the 
 boundary of $D(B_k)$ in $\Lo^{n+1}$. 
As $B_k$ is a disk in the slice $\{t=0\}$,  
 the boundary  of its domain of dependence  is composed by two lightlike cones merged at the boundary of $B_k$ in $\{t=0\}$.    Consistently with \S \ref{ss2.1}, we will drop this last part of the boundary 
 and we will regard  $\parti (\R\times S^{n-1})$  just as a double lightcone (as in the case of $\Lo^{n+1}$), which is also chronologically complete.
  
Now, recall that both, $\R^n$ and $H^n$ are Cartan-Hadamard manifolds (i.e., complete simply connected manifolds with nonpositive curvature). Eberlein and O'Neill \cite{EO} proved that such a manifold $M$ is always endowed  with a natural ideal boundary  at infinity homeomorphic to $S^{n-1}$. Thus, for all these manifolds one would expect $\parti (\R\times M)$ to be also a double cone. 
However, one cannot expect that $\R\times M$ admits a  conformal extension in general. Indeed, starting at $H^n$, many $C^2$ perturbations of the metric will retain its Cartan-Hadamard character and, generically neither the Weyl tensor will be 0 nor it will admit a continuous extension to any  point of the ideal boundary $S^{n-1}$. These properties would be transferred to the product and, being the Weyl tensor conformally invariant,  {\em $\R\times M$  will not admit any continuous conformal extension. } 
\end{example}

In conclusion, the existence of many relevant cases where there exists a well-behaved conformal boundary, makes this boundary the most popular one in General Relativity. However, its generic inexistence as well as the necessity of some ad-hoc additional  hypotheses makes natural to look for a general intrinsic boundary which can be constructed systematically. 

\subsection{Notion of c-boundary}\label{s2b1}
The causal boundary (c-boundary, for short) was introduced by Geroch, Kronheimer and Penrose (GKP) in  \cite{GKP}. 
This  boundary $\partial_c\m$ and the corresponding completion $\m_c=\m \cup \partial_c \m$ is constructed systematically for
any strongly causal spacetime $\m$, starting at a partial future completion $\hat \m= \m \cup \hat \partial \m$ and a dual past one $\check \m$. $\hat \m$ is composed by all the  {\em indecomposable past sets} (IP), which may be
either {\em terminal} (TIP),   
$P=I^-(\gamma)$ , where $\gamma$ is any future-directed timelike continuously inextensible curve,  or {\em proper} (PIP) $P=I^-(p), p\in\m$; TIP's yield $\hat \partial \m$ and PIP's identify $\m$ in $\hat \m$.  For $\hat \partial \m$, analogous future IF's, TIF's and PIF's are defined. To construct $\m_c$, some non-trivial issues appear. They were discussed along decades and carefully analyzed in \cite{FHS11}, whose conclusions are summarized next, stressing the globally hyperbolic case.

\subsubsection{Pairings $(P,F)$}\label{s_Pairings}  
In principle, $\m_c$ is composed of pairs $(P,F)$ of IP's and IF's and, in particular, any $p\in\m$ is represented by $(I^-(p),I^+(p))$ in $\m$.
However, the eventual 
pairing of a TIP $P$ with a $TIF$ $F$ is not obvious. Szabados \cite{Sz88,Sz89} introduced the best option \cite[\S 3.1.1]{FHS11}, namely, given $P\in \hat \m$ (a dual statement is imposed given $F\in \check{\m}$),  the pair $(P,F)$ belongs to $\m_c$ whenever: 
(i) $F\in \check \m$, (ii) $F$ is included in the common future $\uparrow P:=\cap_{q\in P} I^+(q)$ and (iii) $F$ is maximal 
there (i.e. no other $F'\supsetneq F$ satisfies (i) and (ii)); if no such an $F
$ exists then $(P,\emptyset)\in \m_c$. Noticeably, a big simplification appears in the globally hyperbolic case \cite[Th. 3.29]{FHS11}:

\begin{prop}\label{p_timelikepoint_caus}
A strongly causal spacetime is globally hyperbolic if and only if no pair $(P,F)$ with $P\neq \emptyset \neq F$ belongs to $\m_c$.
\end{prop} 

That is, in the globally hyperbolic case, $\partial_c\m$  contains a future and past part, each one identifiable to $\hat\partial \m$ and $\check \partial \m$, resp. (compare with Prop.~\ref{p_timelikepoint_conf}).

\subsubsection{Causality}\label{s_Causality} Once $\m_c$ is defined as a set, it is easy to introduce an extended   {\em chronological relation} $\overline{\ll}$ to $\m_c$, namely, $(P,F)\overline{\ll} (P',F')$ when $F\cap P'\neq \emptyset$. In the globally hyperbolic case, as all the boundary points have either $P$ or $F$ equal to $\emptyset$, the only possible chronological relations appear between points type 
$(\emptyset, F)\in \check \m \subset \partial_c \m$ and  $(P,\emptyset)\in \hat \m\subset \partial_c \m$. 
However, one can still define that, say, $(P,\emptyset)$ and $(P',\emptyset)$ are {\em horismotically related} if $P \supsetneq P'$. In particular, $\{P_t\}_{t\in\R}\subset \hat \partial \m$ is a {\em lightlike line} if $t<t'\Rightarrow  P_t \subsetneq P_{t'}$.

\subsubsection{Topology}\label{s_Topology} To find a reasonable topology for $\m_c$
is much subtler. Indeed,  GKP gave a first approach, but they warned it should  be developed further. 
Beem \cite{Be77} introduced a metrizable topology  applicable to the globally hyperbolic case (see also \cite{CFHadd, Mu19}).   
Harris \cite{Ha00} (see also \cite{Ha98}) introduced the {\em chr-topology} for the partial boundaries $\hat \partial \m, \check \partial \m$. Using it, Flores \cite{Fl07} introduced a topology (also named {\em chr}) in the causal completion $\m_c$. The  analysis in
\cite{FHS11} gave a strong support to this topology, at least as the  minimum  one for $\m_c$, that is, {\em any admissible topology on $\m_c$ should contain (be finer than) the chr-topology}. Indeed, 
Beem's topology fulfills this property. 

The simplification of $\partial_c \m$ in the globally hyperbolic case (Prop.~\ref{p_timelikepoint_caus}) makes the chr-topology in $\partial_c \m$ equivalent to Harris' one in $\hat\partial \m\cup \check{\partial}\m$. 
To describe the chr-topology in $ \hat\partial \m$, first, one defines a limit operator $\hat L$ for sequences $\{P_m\}_m$ of IP's, so that if $P \in \hat L(\{P_m\}_m)$ then $P$ 
is one of the possible limits of the sequence $\{P_m\}_m$. Then, a set $C\subset \hat \m$ is defined as closed  for the chr-topology   if and only if $\hat L(\{P_m\})\subset P$ for any sequence $\{P_m\}_m$  included in $\hat \m$.

Specifically,  $\hat L$ is defined as follows. $P \in \hat L(\{P_m\}_m)$  when it satisfies: (i) $P$ is included in the inferior limit of $\{P_m\}_m$ (for each $p\in P$ exists $k\in\N: p\in P_k$ if $k\geq m$)   and 
(ii) $P$ is {\em maximal} in the superior limit of $\{P_m\}_m$ (which contains the points of $\m$ included in infinitely many $P_m$'s); here, maximal means that  no IP $P'\supsetneq P$ is included in the superior limit.   

\begin{rem} There are some mathematically relevant issues about this procedure. The first one is suggested in the definition of $\hat L$, namely  $\hat L(\{P_m\}_m)$ may contain more than one point and, thus, the chr-topology may be non-Hausdorff (this will be  discussed below), even though it is always $T_1$.  

More subtlety, 
the limit operator may not be of {\em first order}, that is,  $\{P_m\}_m$ might converge to $P$ even if  $P\not\in \hat L(\{P_m\}_m)$. Recall that  the chr-topology is the coarsest one (i.e., the one containing less open subsets)  so that a set,  $C$, is closed whenever it contains the limits pointed out by $\hat L$. However, the construction of such a topology may lead to other limits not pointed out by $\hat L$   (even though the chr-topology remains sequential), see the analysis in  \cite[\S 6]{FHS11}). 
So far, the suggested pathologies  only appear  in maliciously tailored mathematical examples.
\end{rem}

\subsection{The c-boundary of a static product 
}\label{s_cbounc_staticproduct}

\subsubsection{Summary of known results}\label{s_Summary caus bound} Consider a product 
$\m= (\R \times M, -dt^2+\g)$ with complete $\g$ (that is,   globally hyperbolic). Nowadays, its causal boundary is fully well understood. In the case $M$ compact,   $\hat \partial \m$ is just the TIP $i^+=M$, otherwise: 
\ben \item 
\label{ia}
 Harris \cite{Ha01} studied  $\hat\partial \m$ as a pointset proving that there is a natural bijective  correspondence between $\hat\partial \m$ and the space of Busemann-type functions on $M$. Then, $\hat\partial \m$ becomes a lightlike cone on the Busemann boundary $\partial_BM$ with apex $i^+$,  that is, it contains a lightlike line with the common vertex $i^+$ for each point of $\partial_BM$.

\item 
\label{ib}   Flores and Harris \cite{FH07} studied the chr-topology on this spacetime, and  gave a noticeable example, the {\em  unwrapped grapefruit on a stick}, showing that this topology might be non-Hausdorff\footnote{It is worth pointing out that  such a possibility for the causal boundary was known only in the non-globally hyperbolic case. Indeed,  it is very easy to construct examples with two distinct TIP's, $P_1, P_2$, paired with the same TIF,   $F$. The reader can check that this occurs for  the half plane $t>0$ of $\Lo^2$ with the segment $0<t\leq 1, x=0$ removed, being $F$ the chronological future of the point  $t=1,x=0$ in $\Lo^2$.}.     \cite[\S 2.1]{FH07}. 

\item 
\label{ic} The extensive  analysis in \cite{FHS_Memo} showed that the Busemann completion $M_B = M\cup \partial_B M$ is related to a compactification 
$M_G
$ 
introduced by Gromov in \cite{Gr}, so that  $M_B$ is included in $M_G$ in a natural way. 

Moreover,  the three following possible ``pathological'' properties become equivalent: (i)~$M_B\subsetneq M_G$, (ii)~the inclusion $M_B\hookrightarrow M_G$ is not continuous, and (iii)~$M_G$ has ``spurious''  points (which are not reachable by asymtotically ray-like curves). 

When these 
possibilities do not hold, then $M_B=M_G$ and $\hat \partial \m$ becomes a lightcone  on $\partial_BM$, even at the topological level (i.e., so that  the topology is the natural product one when the vertex is removed),   
see the summary in \cite[Figure~6.2]{FHS_Memo}.   
\een

\subsubsection{Brief explanation and an example}\label{s_brief_expalantion}
We will not go into the details of the boundaries mentioned above, but will give  a simple intuitive idea. 

In Riemannian Geometry, given a ray (unit complete minimizing half geodesic) $c:[0,\infty)\rightarrow M$, its Busemann function $b_c: M\rightarrow \R$  is defined  as the limit $b_c(x)=\lim_{t\rightarrow \infty}(t-d(x,c(t))$, for $x\in M$, where $d$ is the distance associated with $\g$.
In the case of Cartan-Hadamard $n$-manifolds, the set of Busemann functions, up to an additive constant, yields the aforementioned Eberlein and O'Neill's boundary equal to $S^{n-1}$. 

In the case of a static product, any future directed  inextensible causal curve parametrized with $t$, $\gamma(t)=(t,c(t)), t\geq t_0$, will have $|\dot{c}| \leq 1$. This permits to define a Busemann-type function $b_c$ by using the same formal expression as above (and including $b_c\equiv \infty$, which can be obtained with a  constant curve $c$). Then, the IP $P=I^{-}(\gamma)$ is the strict hypograph of $b_c$:
$$
I^-(\gamma )= \hbox{hyp}(b_c):=\{(t,x)\in \m: t<b_c(x), x\in M\}.
$$
Thus, there is a bijective correspondence\footnote{Two subtleties in this correspondence are the following: (a) it is irrelevant to impose $|\dot c|\leq 1$ or $|\dot c|< 1$ in the definition of Busemann functions because if $\gamma$ is an inextendible causal curve then $I^-(\gamma)$ is a TIP (see \cite[Prop. 3.32]{FHS11} and \cite[Remark 4.17 (1)]{FHS_Memo}), and (b)~the function  identically equal to $\infty$ is included as a Busemann one, as it corresponds with the Busemann function $b_{c_{x_0}}$ of the curve $c_{x_0}$ constantly equal to any $x_0\in M$
(thus, the timelike curve $\gamma_{x_0}=(t,x_0)$, $t\in\R$ yields the TIP $P=I^-(\gamma_{x_0})=\m=: i^+$); no other Busemann function can reach the value $\infty$ as all of them are Lipschitz.} 
between Busemann type functions and TIP's. Moreover, $\partial_BM$ is defined as the set of the classes of equivalence  $b_c+\R$ each one containing Busemann functions which only differ by an additive constant.  All the TIP's corresponding to the same class $b_c+\R$ will lie in a lightlike line. This  explains the point \eqref{ia} in \S \ref{s_Summary caus bound}.  

To construct the Gromov completion $M_G$ of $M$ one starts  by identifiying  each $x\in M$ with $d(x,\cdot)$ (the function ``distance  to $x$''), which is  Lipschitz,  and noticing that its class  $d(x,\cdot)+\R$ also identifies $x$. $M_G$ is obtained by taking the closure of $\{d(x,\cdot )+\R: x\in M\}$ in the quotient space 
Lip$(M,\g)/\R$. This space is compact\footnote{One can take a representative of each class in Lip$(M,\g)$ which vanishes at some chosen $x_0\in M$ and it is easy to find that the set of Lipschitz function which coincide at a point are compact (recall that the topologies of pointwise convergence and uniform convergence on compact sets agree).} and, then, so is $M_G$ (it is closed set in  a compact one). The inclusion $M_B\subset M_G$ holds taking into account that each Busemann function $b_c$ is Lipschitz, what explains the first assertion in \eqref{ic}, \S \ref{s_Summary caus bound}.

In order to highlight the topological difficulties  in \eqref{ib} and the remainder of \eqref{ic}, we will consider an example introduced in \cite{FHS_Memo}, which is simpler than the aforementioned grapefruit stick. It can be regarded as a surface in $\R^3$ which is topologically a diverging connected sum of infinitely many tori (Fig.~\ref{fig11}). For the purpose of the completions, it is enough to consider the infinite stairs with diverging steps (which is a CW complex rather than  a manifold) depicted  in Fig. \ref{fig22}. As a warning,  the distance between consecutive steps is constant, even though 
they are depicted as if it  decreased towards the left. Gromov's boundary adds a ``last step'' at infinity and, as suggested in the figure, sequences in the stairs diverging towards the left will admit partial subsequences converging to points in this last step. Busemann's boundary, however, only includes two points of Gromov's, see Figure \ref{fig3}. Indeed, the points in the interior of Gromov's  last step are ``spurious'' from Buseman's viewpoint, as they cannot be reached as endpoits of any curve in $M$.  $M_G$ is always a compactification of $M$ and $M_B$ is a sequential compactification; however, Beem's topology would not compactify it (compare with \cite{CFH_II}).

\begin{figure}
	\centering
\includegraphics[height=0.1\textheight]{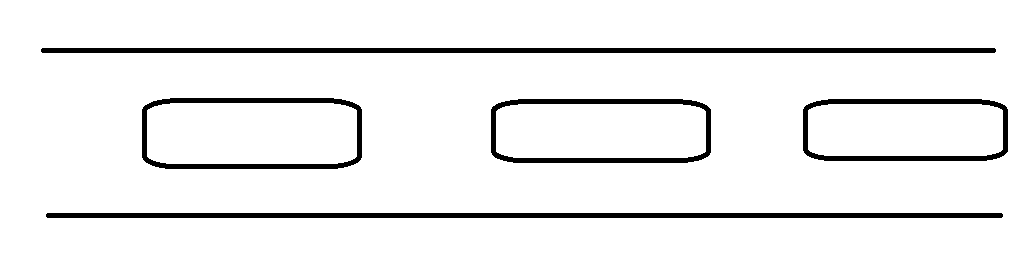}
\caption{\label{fig11} A complete 2-dim example Riemannian manifold with non-Hausdorff $\partial_BM$ homemorphic to the connected sum of infinite tori. $M_B$ and $M_G$ behave qualitatively as the infinite 1-dim stairs in Fig. \ref{fig22}.}
\end{figure}

\begin{figure}
	\centering

\includegraphics[height=0.1\textheight]{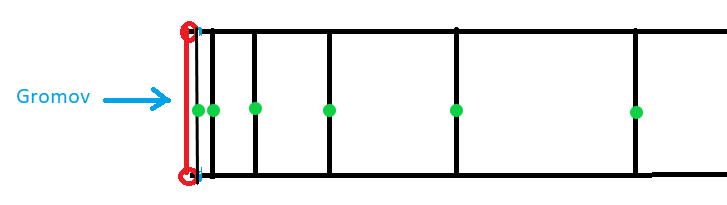}
\caption{\label{fig22}  The steps of the infinite stairs diverge, but we make them to accumulate at the right so that the ``last step'' given by the Gromov boundary can be visualized (in red). The  depicted sequence in the  steps (in green) will  converge to a single point in this last step. }
\end{figure}

\begin{figure}
	\centering

\includegraphics[height=0.1\textheight]{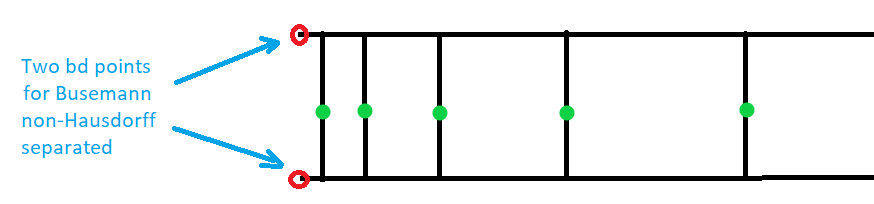}
\caption{\label{fig3} All the  points in $\partial_BM$ must be constructed as Busemann functions for some curve $c$. So, only the two endpoints of Gromov's last step belong to $\partial_BM$, and the depicted sequence converge to both. }
\end{figure}

\begin{rem}\label{r_linksLF1}
(1) The study in \cite{FHS_Memo} includes   the non-globally  hyperbolic case and, thus, the role of the Cauchy boundary of $\g$ in $\partial_c\m$. In short,   each point in that boundary yields a timelike line from $i^-$ to $i^+$ which can regarded as a naked singularity. 

(2) The  case of standard stationary spacetimes was also considered in  \cite{FHS_Memo}. As pointed out in Remark \eqref{r_linksLF0}, the causality of these spacetimes is characterized by the properties of the Fermat metric $F=\sqrt{g_0+ \omega_0^2}+\omega_0$. Indeed, this metric also determines the reversed one 
$$\tilde F(v)=F(-v)=\sqrt{g_0+ \omega_0^2}-\omega_0 
$$
and these two Finsler metrics, $F$ and $\tilde F$, 
yield two different Busemann completions $\partial_B^+M$, $\partial_B^-M$ (as well as two Gromov's) which determine $\hat \partial M$, $\check\partial M$, resp. This reflects the fact that the temporal inversion $t\rightarrow -t$ is not an isometry (thus, $\hat \partial\m$ may not be isomorphic to $\check\partial\m$). 
\end{rem}

\subsubsection{Mild dependence on $t$: GRW spacetimes}\label{s_GRW} Previous examples can be extended to the case of Generalized Robertson-Walker (GRW) spacetimes, which are defined as the warped products 
$$
(I\times M, g^f=-dt^2+  f^2(t)  \g), \qquad I=(a,b)\subset \R, \quad -\infty\leq a < b\leq \infty,
$$
where $f$ is a function on $I$ (see \cite{Sa98} for general background and \cite{AF} for boundaries). The conformal time change $ds=dt/f(t)$ yields 
 $$
\frac{g^f}{f^2(t(s))}=-ds^2 + \g \quad s \in I_0=(a_0,b_0),
$$
where $I_0$ is a new interval. Global hyperbolicity is still equivalent to the completeness of $\g$ and, assuming it,
if  $b_0=\infty$, $\hat \partial M$ is clearly a cone on $\partial_B M$ as in the static product case. However, in the case 
$b_0<\infty$ one has a trivial spacelike conformal boundary  
at  $t_0=b$ diffeomorphic to $M$
which agrees with the causal one $\hat \partial M$ at the topological and causal levels.

It is  worth pointing out that these results admit a non trivial extension to the case of multiwarped spacetimes type $(I\times M_1 \times \cdots \times M_k, -dt^2+\sum_{i=1}^kf_i^2g_{M_i})$, see \cite{AFH}. 

\subsubsection{Conclusion}
Both, the constructive procedure of  the c-boundary of a static product as well as the cases when this boundary is not well-behaved, are fully well understood. 

Under our viewpoint, one should not be specially worried about the mentioned possibility of a bad behaviour. Notice, for example, that the Cauchy boundary of a Riemannian manifold is commonly regarded as the simplest and more natural boundary therein. However,  
such a boundary may have ``bad'' properties \br even in the case of open subsets of $\R^n$, such as being non-locally compact (thus  different to the topological boundary in $\R^n$, see \cite[Example 4.9 (d)]{FHS_Memo}) \er
or having a fractal dimension \br (Koch snowflake). \er

\section{Causal boundary for sliced spacetimes}\label{s_causalboun_for_sliced}

In this section, first, we will discuss  the difficulties to reduce the c-boundary of a globally hyperbolic sliced spacetime  
to the boundary of a static product (which is known from \S \ref{s_cbounc_staticproduct}), 
under uniforms  bounds  in the same vein as in \S \ref{s_unifbound}. Then, we will see how to overcome them.

\subsection{Lowering the expectations on uniform bounds}\label{s_lowering}
One can wonder at what extent the causal boundary of a globally hyperbolic spacetime $\m$ written, up to a conformal factor, as a sliced one \eqref{e_orthogonalnormaliz_spl}, $\m=(I\times M, g=-dt^2+g_t)$, is equal to the boundary of a static product
$\m_0=(I\times M, -dt^2+g_M)$ (with $g_M$ complete), studied in \S \ref{s_cbounc_staticproduct}. Here, $g_M$ will play the role of $g_R$ in Prop. \ref{p_Bari}, so that  one should impose a uniform bound such as \eqref{e_unif_bound_gh} at least. Indeed, being more restrictive,  assume the ``sandwich'' bound: 
\be \label{e_uniform_constant_bounds}
c_1 \, g_M \leq g_{t_0} \leq c_2 \, g_M, \qquad \forall t_0\in I,
\ee
for some constants $0<c_1<c_2$. If $g_i:=-dt^2+c_i \, g_M$, $\m_i:=(I\times M, g_i)$ for $i=1,2$,  this bound implies $g_2\preceq g \preceq g_1$ and, clearly, one can identify naturally $\partial_c\m_0 \equiv \partial_c \m_1 \equiv \partial_c \m_2$. However, the simple example in \S \ref{ss4} (where  all the slices are isometric to $\R$) shows that {\em $\partial_c\m$ may be very different to $\partial_c \m_0$ even under the sandwich bound \eqref{e_uniform_constant_bounds}}.

\subsection{Key of the difficulties: isocausality}\label{s_isocausal}
The previous problem can be well understood thanks to the notion of {\em isocausality}, introduced by Garc\'{\i}a-Parrado and Senovilla in \cite{GPS}.

Two spacetimes $\m$, $\m'$ are called isocausal if there exists two diffeomorphisms, one of them 
 $\phi:\m\rightarrow \m '$ sending the future-directed causal cones of $\m$ inside those of $\m'$
and the other $\psi:\m'\rightarrow \m$ sending the causal cones of $\m'$ inside those of $\m$. Notice that, in general, 
  $\psi\neq \phi^{-1}$, the equality corresponding to the trivial case when $\m$ and $\m'$ are conformally equivalent. As an example, every spacetime $\m$ is {\em locally} isocausal to Lorentz-Minkowski $\Lo^{n+1}$.
  
Most typical causality properties for a spacetime (as being chronological, causal, distinguishing, strongly causal or stably causal) improve when narrowing  cones. So, if $\m$ fulfills one of them, then so will do all its isocausal spacetimes. In particular, this happens for global hyperbolicity, and  underlies  the applicability of the one-sided uniform bound in Prop.~\ref{p_Bari}. 
  
However, this is not the case for all the causal properties. An example in \cite[Example 3.1]{GS} exhibits two isocausal spacetimes, one of them  causally simple and the other not; indeed, it is two steps below in  the classical hierarchy of   causality\footnote{For this hierarchy, see for example \cite{MinSan}.}  as it is  neither causally continuous. 

The fact that  the  c-boundary is {\em not} preserved   by isocausality was noticed in \cite{FHS_isoc}. Indeed, the aforementioned counterexample in \S \ref{ss4} is a variant of one in this reference (which had a timelike boundary and, so, was is not globally hyperbolic).

\subsection{Natural hypotheses to compute the c-boundary}\label{s_computecbound}
 In spite of these difficulties, a closer look shows that isocausality can be used to obtain very valuable information on the c-boundary in the previous setting. A general procedure is developed in  \cite{FHS13}, and we describe briefly it for our case. 
 
 Consider a sliced spacetime $\m= (\R \times M, g=-dt^2+g_t)$, let
$g_M$ be a complete  metric  on $M$ and $\alpha: I\rightarrow \R$ such that  $0<\alpha< 1$, and define $\gc:= -dt^2+g_M$, $\go:= -dt^2+ \alpha^2(t) g_M$. Assume that
 
\begin{equation}\label{e_alpha_bound}
\alpha^2(t) g_M \leq g_t \leq g_M, \qquad \hbox{that is,} \quad \gc\preceq g \preceq \go.
\end{equation}
The key results can be summarized as follows:
\ben
\item \label{(1)} For any $\alpha$, the temporal function $t$ is Cauchy for $\go$  (apply  Prop.~\ref{p_Bari} to $\gc$). In particular, $\go$ (and $g$) is  globally hyperbolic.
\item \label{(2)} If $\int_0^\infty dt/ \alpha(t)=\int^0_{-\infty} dt/ \alpha(t)=\infty$, then $\go$ is conformal to $\gc$ (change to conformal time $ds=dt/\alpha(t)$ as in \S\ref{s_GRW}). Thus, $\m$ is isocausal to the static product $\m_{\tiny{cl}}=(\R\times M, \gc)$.

\item \label{(3)} Assuming that $\alpha$ satisfies the more restrictive hypothesis:
\begin{equation}\label{e_alph_asymp}
\int_0^\infty  \left(\frac{1}{\alpha(t)}-1\right) dt <\infty,
\qquad \int^0_{-\infty}
\left(\frac{1}{\alpha(t)}-1\right) dt <\infty,
\end{equation}
(in particular, if $1-\alpha(t) \sim O(1/|t|^{1+\epsilon})$ for some $\epsilon>0$ and large $|t|$) then there is a natural inclusion of  $j:  \partial_c \m_{\tiny{cl}} \hookrightarrow 
\partial_c  \m$. Thus, $\partial_c  \m$ includes the double cone with apexes  $i^+,i^-$ corresponding to  the Busemann boundary of $(M,g_M)$ (see \cite[Th. 5.5, Cor. 5.6]{FHS13}).
\een 
More precisely about the last point, the inclusion $j$ is defined just as 
$$
j(P_{\tiny{cl}}) = I^-(P_{\tiny{cl}}), 
\qquad \hbox{for any TIP} \; P_{\tiny{cl}}\in \hat\partial \m_{\tiny{cl}}
$$  
(the past $I^-$ computed with $g$), and analogously for TIF's. Here,  \eqref{e_alph_asymp} is used to prove the  injectivity of $j$, \cite[Prop. 5.3 (2)]{FHS13}. 
It is possible, however, the existence of two $g$-TIP's, $P^1,P^2$ such that:
$$P_{\tiny{cl}}\subset P^1\cap P^2 \qquad \hbox{and} \qquad  I_{\tiny{op}}^-(P^1)=I_{\tiny{op}}^-(P^2)=I_{\tiny{op}}^-(P_{\tiny{cl}}).$$ In this case, we say that  $P^1$ and $P^2$ are st-related. 
 Under \eqref{e_alph_asymp}, this is a relation of equivalence\footnote{Lem. 5.10 in \cite{FHS13} provides both, the existence of $P_{\tiny{cl}}$ for reflexivity  and its uniqueness for transitivity in  \cite[Def. 5.11]{FHS13}.},  whose classes are called {\em strains}\footnote{Properly, the name {\em strain} would correspond to the non-trivial classes, i.e., those with more than a point, as in  \cite[Def. 5.11]{FHS13}.}. Then, the map:
$$\hat\partial \m_{\tiny{cl}}\ni P_{\tiny{cl}}\mapsto [j(P_{\tiny{cl}})] \in \{\hbox{strains of} \; \partial \m\}  
$$ is always bijective and continuous. Moreover, it is a homeomorphism when the codomain (the space of all the strains) is Hausdorff   
 \footnote{In this case, the boundary $ \hat  \partial  \m_{cl}$ will be
 Hausdorff too. It is worth pointing out that the non-globally hyperbolic
 case is also considered in \cite[Th. 5.16]{FHS13}. In order to obtain the conclusion
 in this more general setting, one imposes hypotheses so that the Cauchy
 boundary can neither introduce additional non-Hausdorff related points in
 $\hat \partial \m_{cl} $. Then the Hausdorffness of the quotient implies
 that  $\hat \partial \m_{cl}$ is also Hausdorff  (see also \cite[Remark
 5.17]{FHS13}).}  \cite[Th. 5.16]{FHS13}.


\begin{rem}\label{r_linksLF2}
In the same line as Remarks \ref{r_linksLF0}, \ref{r_linksLF1}, it is worth poining out that the previous study of the c-boundary can be extended to the case when the bounding spacetime is standard  stationary instead of a static product (now the Fermat metric $F$ and its reverse $\tilde F$ would play the role of  $g_M$). Indeed, this case as well as  the case when the sliced metric is not globally hyperbolic were carefully considered in \cite{FHS13}. 
\end{rem}

\subsubsection*{Conclusion}  Given a sliced spacetime $\m$, if one finds a static product $\m_0=(\R\times M,-dt^2+g_M)$ with $g_M$ complete and a uniform bound \eqref{e_alpha_bound}  converging asymptotically to the static product as in \eqref{e_alph_asymp},  then $\partial_c\m$ becomes a ``strained'' static double cone. The strains are due to the rather uncontrolled freedom for the cones which is permitted by such bounds on the metric (even if the metric converges uniformly to the static one with $|t|$).
  
This provides an  attractive  picture of $\partial_c\m$ as an ``skeleton'' (see \cite{FHS13}). First, it contains   $\partial_c\m_0$, which serves as the ``column'' (a sort of profile at large scale) of $\partial_c\m$.
  Each one of the points in $\partial_c(\R\times M)$  is represented as $I^-(P_{\tiny{cl}})$ in $\partial_c\m$  and it is  the representative of a strain in the column which is included in $I^-_{\tiny{op}}(P_{\tiny{cl}})$.

Such a procedure may be extended to other cases (see Remark \ref{r_isocaus_paraelfuturo} below) and, eventually,  one could try to control the strains by imposing $C^1$ bounds.  In any case,  direct computations would be possible for specific  physical models (where, probably,  proper strains will not appear).

\subsection{Globally hyp. s.t. with-timelike-boundary }\label{s_ghwithboundary} We end by discussing briefly the general case obtained  taking spacetimes with boundary. Specifically, we consider
globally hyperbolic spacetimes-with-timelike-boundary, which were studied systematically in Sol\'{\i}s' PhD Thesis \cite{Didier}, providing the background for Chrusciel et al. \cite{CGS}. Recently, they have been revisited in \cite{AFS}, which will be followed here (see also Ak\'e's PhD Thesis \cite{Ak}).

\subsubsection{Basic ingredients}\label{s_borde1} The conformal and causal boundary approaches merge naturally in the case of a globally hyperbolic spacetime-with-timelike-boundary $\overline{\m}= \m \cup \partial \m$. In fact, the timelike  boundary $\partial \m$  may have different interpretations such as  a conformal representation of the naked singularities at infinity (as it happens for AdS spacetime in a standard way)
or a cut-off for the asymptotic behaviour of the spacetime.
In the latter case, one would circumvent the subtleties of the c-boundary at infinity, by imposing boundary conditions on the relevant fields at $\partial \m$. Eventually, the cut-off migth be  carried out only in some directions and, in any case, the inclusion $i:\partial \m \hookrightarrow \overline{\m}$ can be seen as a conformal boundary naturally included  in $\partial_c \m$.  

A spacetime-with-timelike-boundary $\overline{\m}$ is a manifold with boundary endowed with a (time-oriented) Lorentzian metric $g$ such that its boundary $\partial \m$ is timelike; in particular, $\partial\m$  becomes a spacetime with the restriction of $g$. The study of its causality mimics the case without boundary. However, the following caution about  regularity 
must be taken into account: causal curves must be taken  locally Lipschitz\footnote{For causal curves this is equivalent to being $H^1$, which is the natural regularity hypothesis to study variationally spacelike geodesics \cite{Masiel, CanSan}.} from the beginning in order to compute the appropriate causal futures and past (even in the case that $\overline{\m}$ is $C^\infty$). Indeed, a striking difference with the case without boundary  is that  piecewise smooth causal curves may reach less points than locally Lipschitz ones (i.e., $J_{ps}(p) \subsetneqq J_{Lipsch}(p)$), see \cite[Appendix B]{AFS} for 
a counterexample.
  
Following the case  without boundary $\overline{\m}$ is called {\em globally hyperbolic} when it is causal and $J^+(p, \overline{\m})\cap J^-(q,\overline{\m})$ is compact\footnote{As pointed out above, $J^\pm$ are computed with locally Lipschitz causal curves. A more restrictive hypothesis would be to assume that  the compactness  holds when $J^\pm$ are computed with piecewise smoooth ones \cite[Prop. 2.22]{AFS}.} for all $p,q\in \m$ (here the notation makes explicit that the futures and pasts  are computed on the whole $\m$). In this case   $\partial \m$ is globally hyperbolic too, and the interior $\m$ is causally continuous \cite[Th. 3.8]{AFS} (both with the restricted metric).
Taking into account Prop. \ref{p_timelikepoint_caus} and the equivalence between the conformal and causal boundaries stated below Prop. \ref{p_timelikepoint_conf} one has:

\begin{prop}\label{p_timelikepoint_borde}\cite[Th. A6]{AFS}
A strongly causal spacetime-with-timelike-boundary $\overline{\m}$ is globally hyperbolic if and only if all the nontrivial pairs $(P,F) \in \overline{\m}, P\neq \emptyset\neq F$ are identifiable to points in $\partial \m$ (i.e. $P=I^-(\bar p, \m), F=I^+(\bar p, \m)$ for some $\bar p\in \partial \m$).
\end{prop}
\subsubsection{Extension of results to the case with boundary}\label{s_borde2}
The techniques in the case without boundary apply now yielding a natural extension of Geroch theorem and, in particular, the equivalence between global hyperbolicity and the existence of an acausal Cauchy hypersurface-with-boundary $\bar S=S\cup \partial S$ (crossed exactly once by any inextensible causal curve in $\overline{\m}$). Moreover,  
the following non-trivial extension also holds.

\begin{thm}\label{t_princ_borde}
\cite[Th. 1.1 and \S 5.2]{AFS} Any  globally hyperbolic spacetime 
with-timelike-boundary $\overline{\m}$ admits a steep  Cauchy temporal function $\tau$ {\em with $\nabla \tau$  tangent to $\partial \m$.}

As a consequence, $\overline{\m}$ is isometric to a Cauchy temporal splitting 
 $(\R\times \bar S, 
g= -\Lambda d\tau^2 + g_\tau)$  
 with $0<\Lambda\leq 1$, where $\bar S$ is a spacelike Cauchy hypersurface-with -boundary, 
  and every $\tau$-slice is also Cauchy and inherits a Riemannian metric given by $\tau$. 
\end{thm}
Here, the specific difficulty comes from the fact that  $\nabla \tau$ must be tangent to $\partial \m$ so that $\partial \m$ remain orthogonal to all the $\tau$-slices. This  is solved in three steps: 

(a) Working as in the case without boundary, one finds any Cauchy temporal function   $\tilde\tau$. As $\tilde\tau$ will remain $C^0$-stable in the set of all the metrics, there will exist another metric with stricly wider cones $g'>g$ such that $\tau$ remains Cauchy temporal for every $g^*\leq g'$.
 
(b) Now, one finds $g^*$ satisfying  $g<g^*<g'$   such that it can be extended smoothly to the double manifold\footnote{i.e., $\overline{\m}^d$ is the manifold without boundary obtained by taking two copies of $\overline{\m}$ and merging the homologous points along the boundary.} $\overline{\m}^d$  so that the natural reflection $r$ in\footnote{$r$ maps each point in one of the copies of $\overline{\m}$ into  its homologous in the other copy.} $\overline{\m}^d$  is a isometry. This is achieved by choosing $g^*$ as a 
product around $\partial\m$.

(c) Take a  Cauchy temporal function $\tau^d$ for $g^*$ on $M^d$
 invariant by the reflection $r$ (recall that  $\tau^d$ can be constructed invariant by any compact group of isometries, see Rem. \ref{r_olaf}). The required $\tau$ is then just the restriction of $\tau^d$ to $\overline{\m}$.

\begin{rem}\label{r_withbound_extensionrepera} Once the splittings extending those in \S \ref{s3} have been achieved, one can wonder for similar extensions of other studied results  in the case without boundary. This can be done with no serious problems, indeed: 

\ben \item[(A)] Consider a sliced spacetime with boundary as in \eqref{e_orthogonalnormaliz_spl} (replacing the manifold $M$ by a manifold with boundary $\overline{M}$). Uniforms bounds as in Prop.  \ref{p_Bari} can be assumed ($g_R$ would be now a complete Riemannian metric on $\overline{M}$) and, clearly, they suffice to ensure that $t$ is Cauchy temporal. 

\item[(B)] The study of the c-boundary in \S \ref{s_causalboun_for_sliced} can be also naturally extended. In particular, 
the three items \eqref{(1)}, \eqref{(2)}, \eqref{(3)} in \S \ref{s_computecbound} remain valid for spacetimes with timelike boundaries, when comparing $\partial_c\m$  with the boundary of a static product-with-timelike-boundary $\overline{\m}_{\tiny{cl}}=(\R\times \overline{M},\gc=-dt^2+g_M)$ with $g_M$ complete. 

This is easy to check for \eqref{(1)} and \eqref{(2)}, but becomes subtler for \eqref{(3)}, as  now the c-boundary  $\partial_c \m$  is not split into $\hat \partial \m$ 
and $\check\partial \m$ (recall that $\partial \m (\subset \partial_c \m)$ is composed of  non-trivial pairs $(P,F)$). However, this case was also considered in \cite{FHS13} and the results therein can still be applied to obtain \eqref{(3)} as in \S \ref{s_causalboun_for_sliced}. Indeed, this is  straightforward when one considers $\partial_c\m$ at the pointset and causal levels, as \cite[Thms. 6.6, 6.9]{FHS13} is applicable\footnote{Recall that the assumption on $d_Q^+$ which appear in that reference is trivially satisfied in a static product. Indeed, it reduces to the fact that the distance $d$ on $M$ associated with $g_R$ extends continuously to  $\overline{M}$.}. 
 At the topological level, \cite[Thm. 6.11]{FHS13} (as well as  previous Lem. 5.13, Thm 5.16 therein for the partial boundaries) is also applicable\footnote{The essential hypothesis for their applicability reduces to imposing that the Cauchy boundary of $(M,g_R)$ is locally compact. However, this is automatically satisfied, as this boundary  is identifiable to $\partial M$.}. 

Summing up, these results for the item \eqref{(3)} assert that, when $\overline{\m}$ admits a uniform bound  by $\overline{\m}_{\tiny{cl}}$ (as in \eqref{e_alpha_bound}   with \eqref{e_alph_asymp}), then $\partial_c \m$ is equal to $\partial_c \m_{\tiny{cl}}$ with eventual ``strains'' at each point. 
$\partial_c \m_{\tiny{cl}}$  can be described in terms of the topological boundary  $\partial (\R\times \overline{M}) = \R\times \partial M$, as $\partial M$ can be identified with the Cauchy boundary\footnote{The Cauchy boundary  is naturally included in the Busemann one $\partial_BM$ but it was not taken into account in \S \ref{s_Summary caus bound}, as the static product therein was globally hyperbolic without boundary.} for $g_M$. Indeed, extending the description in \S \ref{s_Summary caus bound},  $\partial_c \m_{\tiny{cl}}$ is composed by a timelike line for each point in $\partial M$  from $i_-$ to $i_+$ in addition to two cones of lightlines (one starting at $i_-$ and the other ending at $i_+$)  for the remainder of points of the Busemann boundary $\partial_BM$. When $\partial_BM$ is Hausdorff this topology is the natural product one \cite[Figure 6.2]{FHS_Memo}.


\item[(C)] The results can  also be extended to the case when    cross terms type $dt\otimes \omega_t + \omega_t\otimes dt$ are added to the metric splitting, in the spirit of   
Remarks \ref{r_linksLF0}, \ref{r_linksLF1}, \ref{r_linksLF2}. In this case, the 1-form $\omega_t$ is assumed to be defined on the whole $\overline{M}$. Then, $\partial M$ will correspond with points of the  Cauchy boundary for both, the Fermat metric $F$ and its reverse $\tilde F$, so that they will affect to the c-boundary in a similar way as in the case studied above.
\een 
\end{rem}

\subsubsection{Conclusion} \label{s_borde3}
In  any   spacetime $\m$ which is strongly causal (so that its c-boundary is well-defined),   the non-trivial pairs     $(P,F)\in \partial_c\m$,  $P\neq \emptyset \neq  F$   correspond to naked singularities (this may be regarded as a natural interpretation of Prop. \ref{p_timelikepoint_caus}). So, Prop. \ref{p_timelikepoint_borde} states that a spacetime-with-timelike-boundary $\overline{\m}$ is globally hyperbolic if and only if $\partial \m$ contains its naked singularities. In this case,  Th. \ref{t_princ_borde} provides an  evolution of the naked singularities  through the integral curves of $-\nabla\tau/|\nabla\tau|$, which become orthogonal to the Cauchy slicing of the spacetime. The remainder of the c-boundary $\partial_c\m\setminus \partial\m$ behaves essentially as in the case without boundary and  suitable bounds by a static product allows one to understand $\partial_c\overline{\m}$ in terms of the boundary of the latter. 

The   consistency with the case without boundary, including  the  technical simplification of orthogonality for  the naked singularities, 
suggests the possibility  to recover   predictability for globally hyperbolic spacetimes-with-timelike-boundary from the PDE viewpoint. For this purpose, one should prescribe initial data both, at a Cauchy hypersurface-with-boundary $\bar S$ and on the timelike boundary $\partial \m$. This leads to a mixed boundary PDE problem,  as those  studied first by  
K.O. Friedrichs in 1958 \cite{Fr} 
(see \cite{GM} for a recent global study). 
In the last two decades,   many authors  have considered this issue  from different viewpoints \cite{FN, KW, Wi09, Wi12, CV18a,CV18b,EK,Lu15,Lu18,DDF}. So,  this topic has potential to attract interesting future research.


\section{Counterexamples}\label{s6}

Next, let us construct some examples mentioned  along the article\footnote{Other explicit counterexamples involving geometric  properties of globally hyperbolic spacetimes which might have interest for different purposes appear in \cite{MorSan} (infinitely many timelike homotopy classes  between two points) or \cite[\S 6]{MinSan2} (in spite of the existence of solutions to the Lorentz force equation in each timelike homotopy class when  the electromagnetic form is exact, they may not exist when it is only closed).}. We will follow the structure explained by the end of the Introduction, \S \ref{s_Intro}. 

\subsubsection*{Initial framework}\label{ss}
Consider $\R^2$ endowed with several twisted\footnote{Even though not specially relevant for our purposes, some interesting background properties on twisted products are developed systematically in  \cite{PR}.} metrics  type
\begin{equation}\label{e_g}
g=-dt^2 + f^2 dx^2
\end{equation}
where:
\begin{enumerate}
\item $f>0$. So, $g$ is a Lorentzian metric and $t$ a temporal function with the natural time-orientation,
\item 
For each $t_0\in\R$, $f(t_0,x)=1$ for all $x$ outside  a ($t_0$-dependent) compact interval.   Thus, each slice $S_{t_0}:=\{t\equiv t_0\}$ is a spacelike hypersurface which is complete as Riemannian 1-manifold (therefore, isometric to $\R$).
\end{enumerate}
We will  also consider the curves $\gamma_0, \gamma_\pm: [1,\infty)\rightarrow \R$, $\gamma_0(x)=(t_0 (x),x), \gamma_\pm(x)=(t_\pm (x),x)$ and the open subset $U$
satisfying
\begin{equation}\label{e_curvas}
 t_0(x)=-1/x, \qquad t_-(x)<t_0(x)<t_+(x)<0, \qquad \forall x\geq  1, 
\end{equation}
\begin{equation}\label{eU}
\begin{array}{c}
U=\{(t,x): t_-(x)<t<t_+(x), \, \hbox{and} \, 1<x \}. 
\end{array}
\end{equation}

\begin{figure}
\includegraphics[height=0.25\textheight]{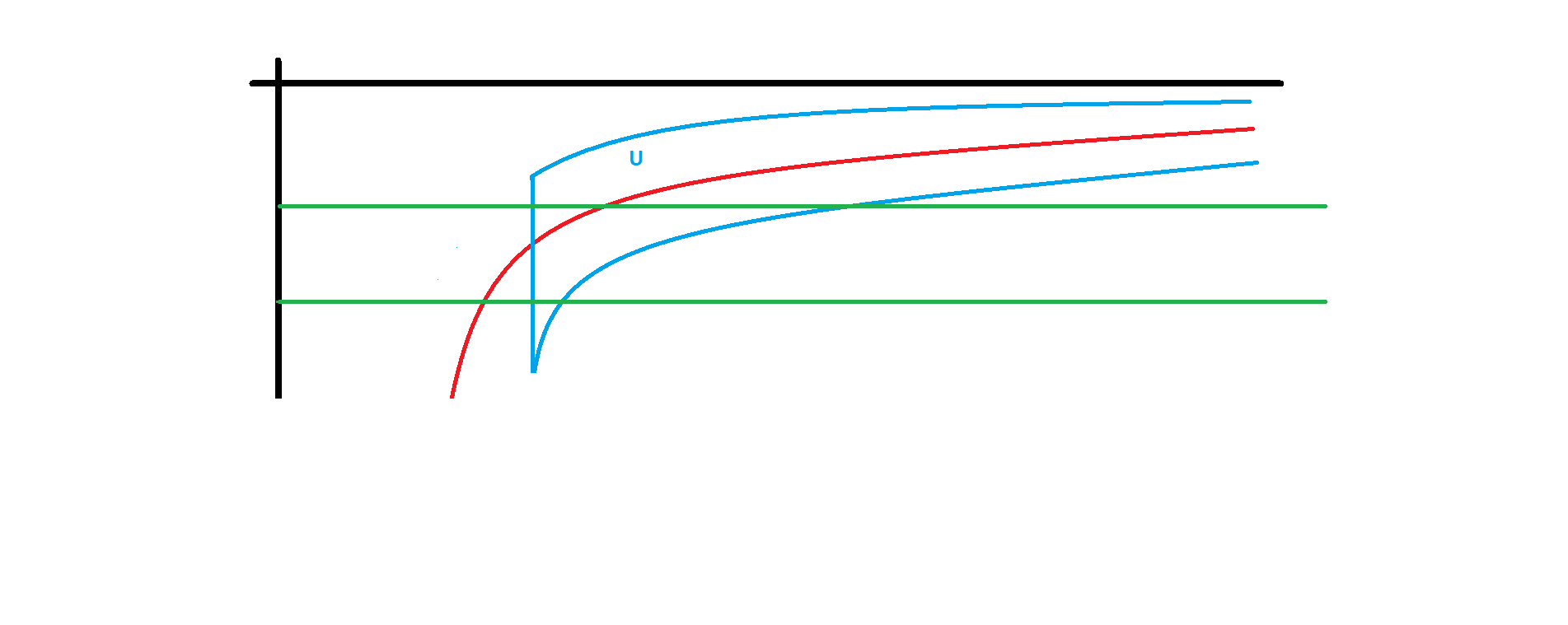}

\caption{\label{fig1} The metric is as in $\Lo^2$ except inside  $U$ (delimited in blue). There, the cones are  opened so that the red curve ($t=-1/x$) becomes lightlike therein. All the $t$ slices (as those in green) intersect $U$ in a bounded interval and are complete, but $t=0$ is not Cauchy.}
\end{figure}

\subsection{Temporal function $+$ complete  slices $\not\Rightarrow$ 
global hyperbolicity}\label{ss1}
Let us construct a metric $g_1=-dt^2+f_1^2(t,x)dx^2$ which is non-globally hyperbolic and  will be the starting point for other examples. Choose:
\begin{equation} \label{e1}
f_1(t,x)=
\left\{
\begin{array}{lll}
1 & & \hbox{if} \; (t,x)\not\in  U, \\
\frac{1}{x^2} & & \hbox{if} \; t=-\frac{1}{x}, \; \hbox{whenever} \; 2\leq x, \\
>0 & & \hbox{otherwise.} 
\end{array}
\right.
\end{equation}
(see Fig. \ref{fig1}). Notice that, $\gamma_0|_{[2,\infty)}$ is a lightlike curve inextensible to the future and, so, the slices $S_{t_0}$ with $t_0\geq 0$ cannot be Cauchy hypersurfaces. However, each slice $S_t$, $t\in\R$ is  complete, as explained above.

\begin{rem} \label{r0} The following two improvements can be obtained directly in this example and the subsequent ones.
\ben\item {\em Ensuring non-hyperbolicity.}
The fact that $t$ is not Cauchy temporal  does not imply directly that  global hyperbolicity fails (a different Cauchy slicing might exist). However, this is guaranteed by the following small modification: {\em choose $f_1$ as above only for $t\leq 0$ and put $f_1(t,x)=f_1(-t,x)$ for $t>0$.} 
Indeed,  the corresponding spacetime would have non-compact $J^+(-1/2,2) \cap J^-(1/2,2)$. 

\item {\em Existence of analytic examples}.
Our example  is $C^\infty$, but analytic ones must exist. The key is to use the stability of the involved   estimates and properties\footnote{This  also applied to construct analytic Cauchy temporal functions in \S \ref{ss_analytic}.}. Indeed, 
the crucial
completeness of each $S_t$ for $g_1$  is a $C^0$-stable property\footnote{\label{f_Lerner2} The $C^0$-stability of metric completeness and incompleteness in the space of all the Riemannian metrics is straighforward and, then, so is the stability of geodesic completeness in the Riemannian case (as a difference with the Lorentzian one, see footnote \ref{f_Lerner}).}.
For the other examples below, notice that all the temporal or Cauchy temporal functions are also temporal or Cauchy temporal for
all the Lorentzian metrics in a $C^0$-neighborhood (see Remark \ref{r_narices} and \S\ref{s_Festability}).

Concretely, we can approximate the function $f_1$ above by an analytic one $h$ which retains its essential qualitative properties, namely: $h \geq  1/2$ outside $U$, $h \geq  1/x^2$ if $t=-1/x$ and $h>0$ everywhere.
This can be achieved directly by applying the 
Whitney $C^k$-fine approximation by analityc functions\footnote{See also further background in \cite{Azagra}. Of course, the posterior Grauert's criterion used in \S \ref{ss_analytic} also works.} \cite[Lemmas 6, 7]{W}: {\em for every $C^k$ function $f: \R^n \rightarrow \R^m$ and every continuous $\epsilon: \R^n \rightarrow (0,\infty)$ there exists a real analytic function $h$ such that
$\parallel D_j h(x) - D_j f(x)
\parallel 
\leq 
\epsilon(x)$ for all 
$x \in \R^n$ 
and $j = 0,1,...,k,$} (where
$D_j$ denotes derivatives of order $j$ and $\parallel \cdot \parallel$ the natural norm induced from $\R^n$).
\een
\end{rem}

\subsection{Cauchy splitting $-dt^2+g_t \not\Rightarrow -dt^2+(1/2)g_t$ glob. hyp.}\label{ss2}

 The following  metric 
 $g_2=-dt^2+f_2^2(t,x)dx^2$ will be a particular case of the previous  $g_1$ satisfying that, for each  $\alpha>1$, the new metric 
\begin{equation} \label{e_g2alpha} 
 g_2^\alpha:=-dt^2+\alpha^2 f_2^2(t,x)dx^2
 \end{equation} 
 is globally hyperbolic with $t$ Cauchy temporal. In particular, $g_2^{\alpha=\sqrt{2}}$ becomes the  counterexample searched in this subsection (up to the modifications in Remark \ref{r0}).
  Concretely, first
 consider the particular choice: 
\begin{equation}
\label{e_tpm}
\gamma_\pm(x)= (t_\pm(x),x) \quad \hbox{with} \quad t_\pm(x)=-\frac{1}{x} \pm \frac{e^{-x}}{x}, \quad \forall x\geq 1.
\end{equation}
Now, assume  the following  strengthening of\eqref{e1} for $f_2$ (the expression of $U$ is just the particular case of \eqref{eU} for our choice of $\gamma_\pm$):

\begin{equation} \label{e3}
\begin{array}{l}
U =\{(t,x): -e^{-x}<tx+1<e^{-x}, 1<x\} \\ 
\\
\left\{
\begin{array}{lll}
f_2(t,x)=1 & & \hbox{if} \; (t,x)\not\in  U,  \\
\frac{1}{x^2} \leq f_2(t,x) \leq 1 & & \hbox{if} \; (t,x)\in  U,
\\
f_2(t,x)=\frac{1}{x^2} & & \hbox{if} \; t=-\frac{1}{x}, \; \hbox{whenever} \; 2\leq x,
\end{array}
\right. 
\end{array} \end{equation}
 Thus, for each vertical line $x=x_0 \geq 2$, one has the cones of Lorentz-Minkowski space (up to a rescaling in the $x$ coordinate, i.e., $-dt^2+d\bar x^2$ with $\bar x=\alpha x$) outside the interval limited by $\gamma_\pm(x_0)$ and, inside this interval, the cones are opened and become wider at $\gamma(x_0)$.

\begin{figure}
	\centering
\includegraphics[height=0.5\textheight]{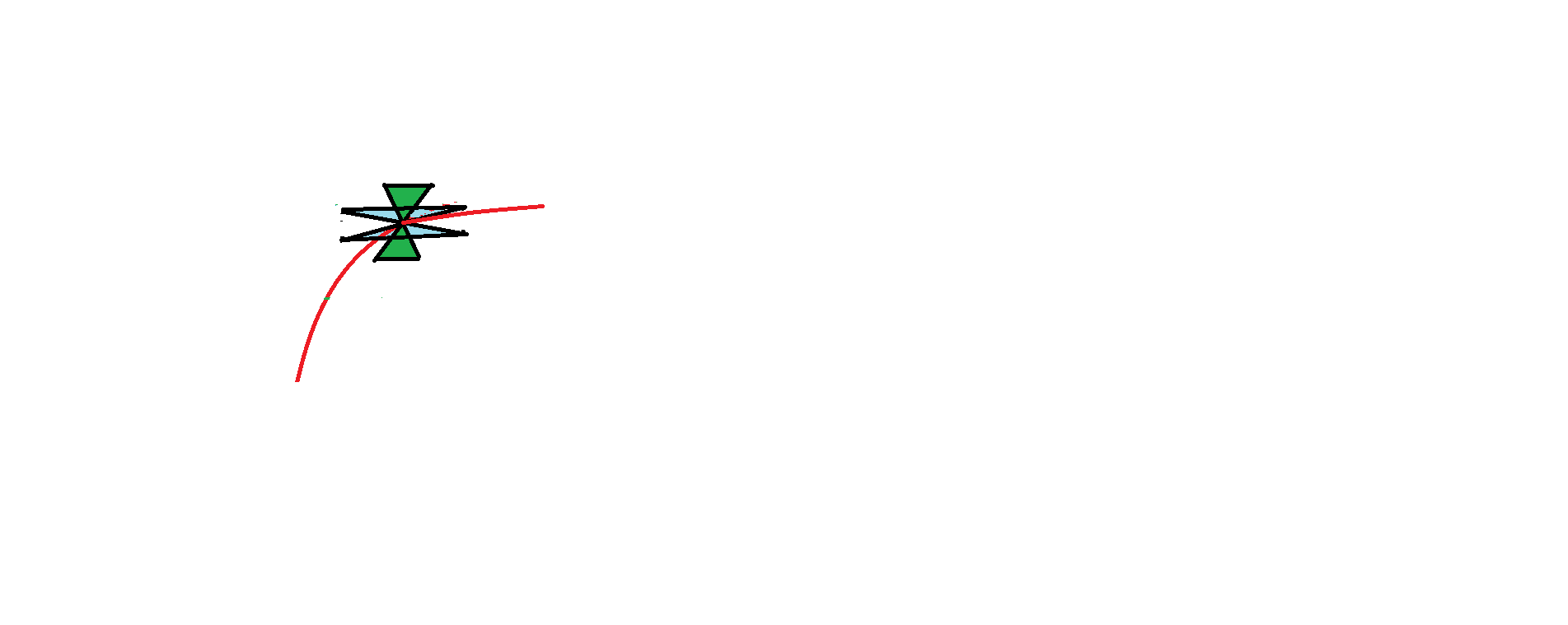}

\caption{\label{fig2} Modification of the example in Fig \ref{fig1}. The narrow cone in green is associated with the globally hyperbolic metric $
g_2^\alpha=-dt^2+\alpha^2f^2_2dx^2$ with $ \alpha=\sqrt{2}$. The wide cone in blue is associated with $g_2^\alpha$ for $\alpha=1$, which is not globally hyperbolic (as its cones behave as in Fig. \ref{fig1}).}
\end{figure}

Thus, the problem is reduced to prove that   $t$ is Cauchy temporal for $g^\alpha_2$ with $\alpha>1$. Then, it is enough to check the following property:

\begin{quote}
(P) Any inextensible 
causal curve $\beta$ starting at $(t_0,x_0)\in U$  will leave $U$ both towards the future and towards the past, and it will not return to $U$, for any $g^\alpha_2$ with $\alpha>1$.
\end{quote}
Indeed, in this case $\beta$ will cross all the slices $S_t$ with $t>0$ and $t<-2$ (recall that $g^\alpha_2$ is  Lorentz-Minkowski in  $\R^2\setminus U$) and, by continuity, all the slices $S_t$ will be Cauchy. 

In order to prove (P), the following simplifications will be used:
\begin{enumerate}

\item \label{i1} Check (P) only for each
$\beta$ either inextensible to the future or to the past totally included in the region \begin{equation} \label{e_des0}
x>a \qquad \hbox{for some} \, a>0.
\end{equation}

Indeed, the region $U\cap \{x\leq a\}$ has compact closure and no such a $\beta$ can be partially imprisoned in it (this is a general property for strongly causal metrics, in particular $g_2^\alpha$, which is stably causal).

\item \label{i2} Check the previous property only for $\beta$ inextensible to the future. 

Indeed, assume $\beta$ is past inextensible. When its $t$-coordinate reaches any value $t_0<0$ then $\beta$ remains in the region $t<t_0$ (as $t$ is a time function) an it cannot remain partially imprisoned in  $\{t\leq t_0\}\cap U$ because  the closure of this set is compact (recall that $\gamma_-$ is asymptotic to the $x$-axis).

\item \label{i3}  
Check only that, whenever $(t_0,x_0)\in \bar U$ (the closure of $U$) with
\begin{equation}\label{e_des}
x_0> -\log (\alpha-1) \qquad \hbox{and} \qquad x_0>1+\frac{2}{\alpha}\; , 
\end{equation}
the two lightlike  $g^\alpha_2$-geodesics $\rho_+,  \rho_-$ (up to reparametrization) starting at  $(t_0,x_0)$ 
must cross $\gamma_+$ at a (necessarily unique) point. 

Indeed,  if this is proved, $\gamma_+$, 
$\rho_+$ and $ \rho_-$ will give the sides of a closed piecewise smooth curve (a curved triangle) 
which encloses a compact region $R_+$. 
The future-directed curve $\beta$ (which  can be assumed to lie in the region \eqref{e_des0} with $a$ bigger than the   right-hand sides of \eqref{e_des}) is included initially in $R_+$. However,    $\beta$  cannot be imprisoned in $R_+$ and it will leave $R_+$ crossing $\gamma_+$. As $\gamma_+$ is  spacelike in the region \eqref{e_des0}, it is  acausal 
too\footnote{\label{foot} Even though a particular proof adapted to our case would not be difficult (see footnote \ref{foot2}),  a  property  with interest in its own right is: {\em for any Lorentz metric $g$ on $\R^2$, every spacelike curve $\bar \beta$ must be acausal}. \\ Recall first that such a $g$ must be stably causal \cite[Th. 3.43]{BEE}. So, any inextensible lightlike geodesic $\rho$ will be diverging and (using Jordan's theorem after the compactifying $\R^2$ by one point), $\R^2 \setminus$ (Im $\rho$) will consist of two connected 
parts. Moreover, such a $\rho$ will be maximizing, i.e. it will not have  cut points (this follows from \cite[Prop. 3.42]{BEE} taking into account that any piecewise smooth causal curve which is 
not a pregeodesic can be deformed into a timelike curve with the same endpoints). 
\\ Now, assume  by contradiction  that a spacelike curve $\bar \beta$ and a causal one $\beta$  join two distinct points $p, q$. We can assume that $\beta$ is timelike because, otherwise,  $\beta$ could be assumed to be a lightlike pregeodesic and a contradiction with the inexistence of its cut points for $-g$ appears. So, let $\rho$ be the 
inextensible lightlike geodesic starting at $p$ satisfying that $\bar \beta'(0)$ and $\beta'(0)$ lie in different connected parts of $T_p \R^2\setminus $ Span$(\rho'(0))$. Initially $\bar \beta$ and $\beta$ lie in different connected parts of $\R^2 \setminus$ (Im $\rho$) and, as $\beta$ and  $\bar \beta$ will meet at $q$, one of these two curves must leave its connected part and meet the lightlike curve
  $\rho$ 
  at a second point $q' \in$ Im$(\rho)$. If  $\beta$ is the leaving curve, $\rho$ will have a cut point  (or directly \cite[Prop. 3.42]{BEE} fails).  Otherwise, the same contradiction appears for $-g$.}, and  $\beta$ cannot come back to $U$. 

\item\label{i4}
Check the previous property replacing $g^\alpha_2$ on $\bar U$ by the simpler one:  $$g^\alpha:=- dt^2+  \frac{\alpha^2}{x^4}dx^2$$

 Indeed, the cones of this metric are wider than those of $g^\alpha_2$ (recall \eqref{e3}) and any $g^\alpha_2$-causal curve will be also $g^\alpha$ causal. So, if the region $R'$ corresponding to $g^\alpha$ is compact, the original region $R\subset R'$ will be compact too. 
\end{enumerate}
Accordingly, let us prove the third item 
using the metric $g^\alpha$ in the fourth one. 
The initial velocities of $\rho_\epsilon$, $\epsilon=\pm 1$, can be chosen: 
 $$\rho_\epsilon(0)=(t_0,x_0),
 \qquad \rho'_\epsilon(0)=(\frac{\epsilon \alpha}{x_0^2}, 1) 
 .$$
 The reparametrization $\bar\rho_\epsilon$ of $\rho_\epsilon$ with the $x$ coordinate is then: 
$$
\bar\rho_\epsilon(x)=\left(t_0+ \epsilon\alpha \left(\frac{1}{x_0}-\frac{1}{x}\right), x\right), \qquad 
\hbox{for} \left\{ \begin{array}{ll}
x\geq x_0, \; \hbox{if} \; \epsilon=1, \\
x\leq x_0, \; \hbox{if} \; \epsilon=-1.
\end{array} \right.
$$
Thus, $\bar\rho_\epsilon  \cap \gamma_+$ is obtained by equating:
$$
t_0+ \epsilon \alpha \left(\frac{1}{x_0}-\frac{1}{x}\right)=-\frac{1}{x} + \frac{e^{-x}}{x},
$$
that is,
$
(\epsilon\alpha + t_0x_0) x - (\epsilon\alpha-1)x_0= x_0 e^{-x}
$, or   
\begin{equation} \label{e10}
F(x)=0, \quad \hbox{where} \quad F(x):=(\epsilon\alpha-1) (x-x_0)-x_0e^{-x}+(1+t_0x_0)x.
\end{equation}
Recall that, $F(x_0) =  x_0(-e^{-x_0}  +1+t_0x_0)$ and, by the expression of $U$ in \eqref{e3}, 
 $F(x_0)<0$. Thi yields the required result for $\rho_\epsilon$: 

(a) Case $\epsilon=1$:
$$\begin{array}{rl}
\lim_{x\rightarrow \infty}F(x) & 
=\lim_{x\rightarrow \infty}\left(
 (\alpha-1) (x-x_0)+(1+t_0x_0)x\right) 
 \\ 
 &  \geq 
\lim_{x\rightarrow \infty}
 \left((\alpha-1) (x-x_0)-e^{-x_0}x\right)=\infty
\end{array}
$$ 
(use  \eqref{e3} again for the inequality, and the first inequality in \eqref{e_des} for the last limit). Thus, there is  a solution of the equation \eqref{e10} in $(x_0,\infty)$. 

(b) Case $\epsilon=-1$:
$$\begin{array}{rl}
F(1) & =(\alpha+1)(x_0-1)-x_0e^{-1}+(1+t_0x_0)>
(\alpha+1)(x_0-1)-x_0-e^{-x_0}
\\
& >
(\alpha+1)(x_0-1)-(x_0+1)=\alpha x_0-(\alpha+2)>0.
\end{array}
$$
(use  \eqref{e3}  for the first inequality, and the second inequality in \eqref{e_des} for the last inequality) 
and a solution appears in\footnote{\label{foot2} In both cases, we know that the required solution of \eqref{e10} must be unique by footnote~\ref{foot}. Anyway, this can be checked directly here using $F'(x)=(\epsilon \alpha-1)+x_0e^{-x}+(1+x_0t_0)$. When $\epsilon=1$, $F'(x)>(\alpha-1)+x_0 e^{-x}-e^{-x_0}>x_0 e^{-x}$ (the latter by the first inequality in \eqref{e_des}) and $F$ is strictly increasing 
on $x_0<x$.  When $\epsilon=-1$, $F''(x)=-x_0 e^x <0$; so, once attained its absolute positive maximum in $[1,x_0)$, $F$ will decrease strictly until $F(x_0)<0$.}
 $(1,x_0)$.

\subsection{Metrics sharing a Cauchy temporal splitting are not convex}\label{ss3}
Our aim will be to construct two metrics $g_3^{\hbox{\tiny{even}}}, g_3^{\hbox{\tiny{odd}}}$ as in \eqref{e_g} admitting $t$ as a Cauchy temporal function such that 
$$\lambda g_3^{\hbox{\tiny{even}}}+ (1-\lambda) g_3^{\hbox{\tiny{odd}}} \quad \quad \lambda\in [0,1]$$
does not admit $t$ as a Cauchy temporal function for some $\lambda\in (0,1)$. By using the modification in Remark \ref{r0}, one can ensure also that the counterexample  is neither globally hyperbolic (and can be chosen analytic).
More precisely, the counterexample will appear for $\lambda=1/2$, and we will prove:
$$ g_2 (\equiv g_2^{\alpha=1}) \prec \frac{1}{2}  g_3^{\hbox{\tiny{even}}}+ \frac{1}{2}g_3^{\hbox{\tiny{odd}}} 
$$
on $\gamma_0(x)=(-1/x ,x)$  for large $x>0$ (recall \eqref{e_curvas}), so that this curve remains causal for the right-hand sum and the slice $t=0$ is not Cauchy   (as happened for $g_2$). Notice that this sum metric has the same cones as the homothetic one 
$g_3^{\hbox{\tiny{even}}} + g_3^{\hbox{\tiny{odd}}}$, 
but the fact that its cones are  wider 
will become apparent because our construction will yield:
\begin{equation}
\label{e_f}
\frac{1}{2} (f_3^{\hbox{\tiny{even}}})^2+ \frac{1}{2}(f_3^{\hbox{\tiny{odd}}})^2 \leq  \frac{3}{4} f_2^2  < f_2^2, \quad  
\hbox{for large} \, x>2. 
\end{equation}
With this aim, we will construct first an auxiliary metric type:
\begin{equation}
\label{g_3}
g_3:= -dt^2 +\alpha^2(x) f_2^2(t,x) dx^2.
\end{equation}
where $f_2$ fulfills \eqref{e3}, i.e., $g_3$ generalizes $g_2^\alpha$ allowing $\alpha$ to depend on $x$. 
 We will use the following technical  property of the studied metrics $g^\alpha_2$. 
 
 \begin{lemma}\label{l1} Let $g_2^\alpha$ as in \eqref{e_g2alpha} with  $\gamma_\pm$ as in \eqref{e_tpm} and assume $\alpha>1$ . 
 
 For each $\tilde x>0$ there exists $ \tilde x < \tilde y^- < \tilde y < \tilde y^+$   such that:
$$J_{g_2^\alpha}^+(\gamma_-(\tilde y )) \cap \bar U \subset \R\times [\tilde y^- , \tilde y^+],$$ 
  where $J_{g_2^\alpha}^+(\gamma_-(\tilde y ))$ is the $g_2^\alpha$-causal future of $\gamma_-(\tilde y )$. 
 \end{lemma} 

\begin{proof}
Recall the construction of $R_+$ below \eqref{e_des} and choose $(t_0,x_0)=\gamma_-(y)=(t_-(y),y)$ therein. The $x$-coordinate of $J_{g_2^\alpha}^+(\gamma_-(y ))\cap \bar U$ lies in $[y^-,y^+]$ where $y^-,y^+$ are the $x$-coordinates of $\gamma_+ \cap \rho_-$, $\gamma_+ \cap \rho_+$, resp., and they can be regarded as a function of $y$. So, we have just to prove that, when $\tilde{x}>0$ has been prescribed, there exists $\tilde{y}$ such that  $\tilde{x}<\tilde y^- (:=y^-(\tilde{y}))$.  Consider the point $\gamma_+(\tilde x)=(t_+(\tilde x), \tilde x)$. As $t_+(\tilde x)<0$ and $\gamma_-$ is asymptotic to the $x$-axis, there exists $\tilde y>\tilde x$ such that $t_+(\tilde x)<t_-(\tilde y)$. Necessarily, $t_-(\tilde y)<t_+(\tilde y^-)$, thus, $t_+(\tilde x)< t_+(\tilde y^-)$ and, as $t_+$ is strictly increasing, $\tilde x< \tilde y^-$.
\end{proof}

\begin{rem}\label{r_basic_escape}
Notice that the segment $[t_-(\tilde y), t_+(\tilde y)] \times \{\tilde y\}$ is included in  
 $J_{g_2^\alpha}^+(\gamma_-(\tilde y )) \cap \bar U$  and  separates it into two connected pieces.  No  future-directed causal curve in $U$ can go from a point with $x<\tilde y^-$ to a point with $x> \tilde y^+$ because  it would reach 
$J_{g_2^\alpha}^+(\gamma_-(\tilde y )) \cap \bar U$ and, then,  either would remain there or  would abandon $U$.
\end{rem}
We will use the following application of the previous lemma.
 \begin{lemma}\label{l2} Consider the metric $g_2^\alpha$ for $\alpha^2=5/4$ $(>1)$. 
 
 There exists two interwined sequences $\{x_n\}_n \, , \; \{y_n\}_n   \nearrow \infty, 
 $
 and two sequences $I_n$, $J_n$ of disjoint intervals,
 $$
I_n=[x^-_n,x_n^+], \quad J_n=[y^-_n,y_n^+], \qquad 2< x^-_n < x_n < x_n^+ < y^-_n < y_n < y_n^+ < x^-_{n+1}, 
 $$
 such that, for all $n\in \N$:

$J_{g_2^\alpha}^+(\gamma_-(x_n )) \cap \bar U \subset \R\times [ x^-_n ,  x^+_n], \quad 
J_{g_2^\alpha}^+(\gamma_-(y_n )) \cap \bar U \subset \R\times [ y^-_n ,  y^+_n].$  

\end{lemma}

\begin{proof}
Choose any $I_1$ and apply inductively Lemma \ref{l1} ensuring $x_n>n$.
\end{proof}

\begin{prop}
Consider a metric $g^*_3$ as in \eqref{g_3}, \eqref{e3} with $\alpha$ satisfying:
\begin{enumerate}
\item $\frac{1}{4} \leq \alpha^2(x) \leq \frac{5}{4},  \forall x\geq 1$ and $\alpha^2(x)=1, \forall x\leq 1. $
\item $\alpha^2 \equiv 1/4$ on $I_n$.
\item  $\alpha^2 \equiv 5/4$ on $J_n$. 
\end{enumerate}
Then, the function $t$ is Cauchy temporal for $g^*_3$ (see Fig. \ref{fig222}).
\end{prop}

\begin{proof}
 As a first observation,  the metric $g_3^*$  outside $U$ is equal to $-dt^2+\alpha^2(x)dx^2$. On $\R^2$, this  is isometric to $\Lo^2$; indeed, the change $\bar x(x)= \int_0^x \alpha(x')dx'$ yields a global isometry with 
$\Lo^2$ (as $\alpha(x)\geq 1/4$ the range of $\bar x$ is the whole $\R$). So, the background of the previous examples still hold. 

Now, notice  that $J_{g_2^\alpha}^+(\gamma_-(y_n ))=J_{g_3^*}^+(\gamma_-(y_n )) $ (for $\alpha^2=5/4$) and $g_2^\alpha=g_3^*$ therein. Let $\beta$ be a future-inextensible $g^*_3$-causal curve starting at a point of $\bar U$. This curve must leave $U$ because, otherwise  it cannot remain imprisoned in any compact set $\bar U \cap (\R\times [1,y_n])$. Thus, it will enter at some $J_{g_3^*}^+(\gamma_-(y_n ))$ and will leave $U$ therein (Remark \ref{r_basic_escape} applies). 
If $\beta$ were past-inextensible, the same argument as in the item \ref{i2} of the example in \S\ref{ss2} applies.
\end{proof}

\begin{figure}
	\centering
\includegraphics[height=0.38\textheight]{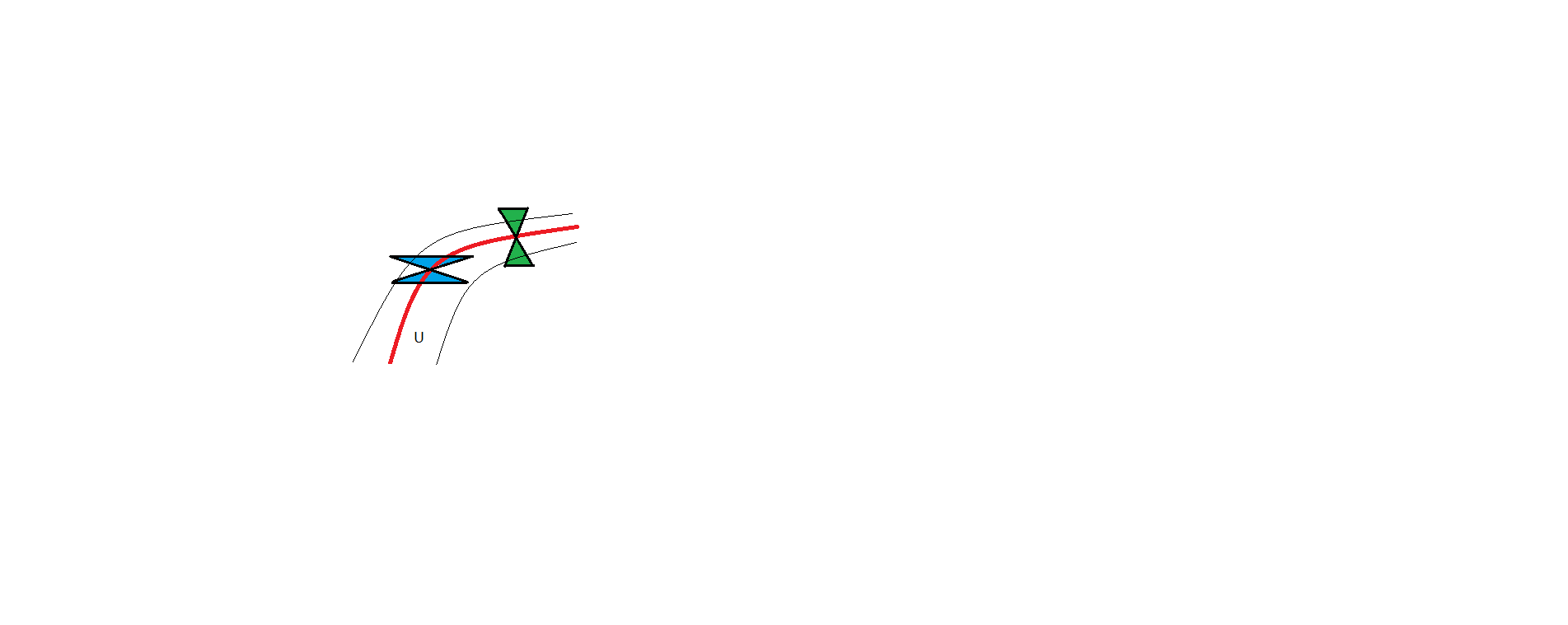}
\caption{\label{fig222}  Function $t$ is Cauchy for the metric $g_3^*$, even if its cones at $U\cap (\R\times I_n)$ are opened so that  $\gamma_0$ becomes timelike. Indeed, such a causal curve is forced to abandon $U$ when crossing the region $\R\times J_n$. The metrics $g_3^{\hbox{\tiny{even}}}$ and $g_3^{\hbox{\tiny{odd}}}$ will have wider and narrower cones in alternate regions, so that each one is globally hyperbolic but  $\gamma_0$ becomes timelike for their semisum. }
\end{figure}

\begin{defi}
Given $g_3^*$ the metrics $g_3^{\hbox{\tiny{even}}}$ and $g_3^{\hbox{\tiny{odd}}}$ are defined as follows. 

\begin{enumerate}
\item $g_3^{\hbox{\tiny{even}}} \equiv g_3^*$ on $[x_{2n}^-,x_{2n+1}^+]$. In particular, in $I_{2n}, J_{2n}$ and $ I_{2n+1}$.

$g_3^{\hbox{\tiny{even}}} \equiv 1/4$ on  $[x_{2n-1}^-,x_{2n}^+]$. In particular, in $I_{2n-1}, J_{2n-1}$ and $ I_{2n}$ (notice that this agree with $g_3^*$ on $I_{2n}$). 

\item $g_3^{\hbox{\tiny{odd}}} \equiv g_3^*$ 
on  $[x_{2n-1}^-,x_{2n}^+]$. In particular, in $I_{2n-1}, J_{2n-1}$ and $ I_{2n}$. 

$g_3^{\hbox{\tiny{odd}}} \equiv 1/4$ 
on $[x_{2n}^-,x_{2n+1}^+]$. In particular, in $I_{2n}, J_{2n}$ and $ I_{2n+1}$ (notice that this agree with $g_3^*$ on $I_{2n}$).
\end{enumerate}

\end{defi}

\begin{thm}\label{t_nonconvex}
The metrics $g_3^{\hbox{\tiny{even}}}$ and $g_3^{\hbox{\tiny{odd}}}$ admit $t$ as a Cauchy temporal function but its convex combination $g_3^{\hbox{\tiny{even}}}/2 + g_3^{\hbox{\tiny{odd}}}/2$ does not admit $t$ as a Cauchy temporal function.
\end{thm}

\begin{proof}
The first assertion follows because both metrics can be regarded as a $g_3^*$, relabelling the subsequences of the even and odd terms.
For the second, notice that, in the points with, say $x\geq x_2^+$, whenever $\alpha^2(x)>1/4$  for one of the metrics (necessarily  then $\leq 5/4$), the other one satisfies $\alpha^2(x)=1/4$. Thus, the semisum is $\leq 3/4$ and  \eqref{e_f} holds, 
as required. 
\end{proof}

\subsection{Normalized Cauchy  splitting with an incomplete slice}\label{ss0}

Let us construct a metric $g_4=-dt^2+f_4^2(t,x)dx^2$ such that $t$ is Cauchy temporal but $t=0$ is incomplete.
For each integer $m\geq 2$, consider the compact region $\hat T_m$ delimited by the polygon (double equilateral triangle)   with vertices $(0,m), (0, m+1), (\pm \sqrt{3}/2,m+1/2)$ 
and a similar one inside 
$T_m$ with  vertices in the $x$-axis 
$(0,m+\frac{1}{2m^2}), (0, m+1-\frac{1}{2m^2})$, 
see Fig. \ref{fig0}.
Now, choose any $f_4$ satisfying:
\begin{equation} \label{e0}
\left\{
\begin{array}{lll}
f_4(t,x)=\frac{1}{1+x^2} & & \hbox{if} \; (t,x)\in  T_m, \\
\frac{1}{1+x^2} \leq f_4(t,x) \leq 1  & & \hbox{if} \; (t,x)\in  \hat T_m\setminus T_m, 
\\
f_4(t,x)=1 & & \hbox{otherwise.} 
\end{array}
\right.
\end{equation}
Clearly, $t$ is a temporal function and no inextensible causal curve can be imprisoned in any $\hat T_m$. Being unaltered the metric of $\Lo^2$ in the regions outside the $\hat T_m$'s, all the slices of $t$ must be Cauchy hypersurfaces. However, the length $\ell$ of the positive $x$-semiaxis satisfies:
$$
\ell \leq 2+ \sum_{m=2}^\infty \frac{1}{m^2} \; + \;  \int_0^\infty \frac{dx}{1+x^2}<\infty
$$

\begin{figure}
	\centering
\includegraphics[height=0.35\textheight]{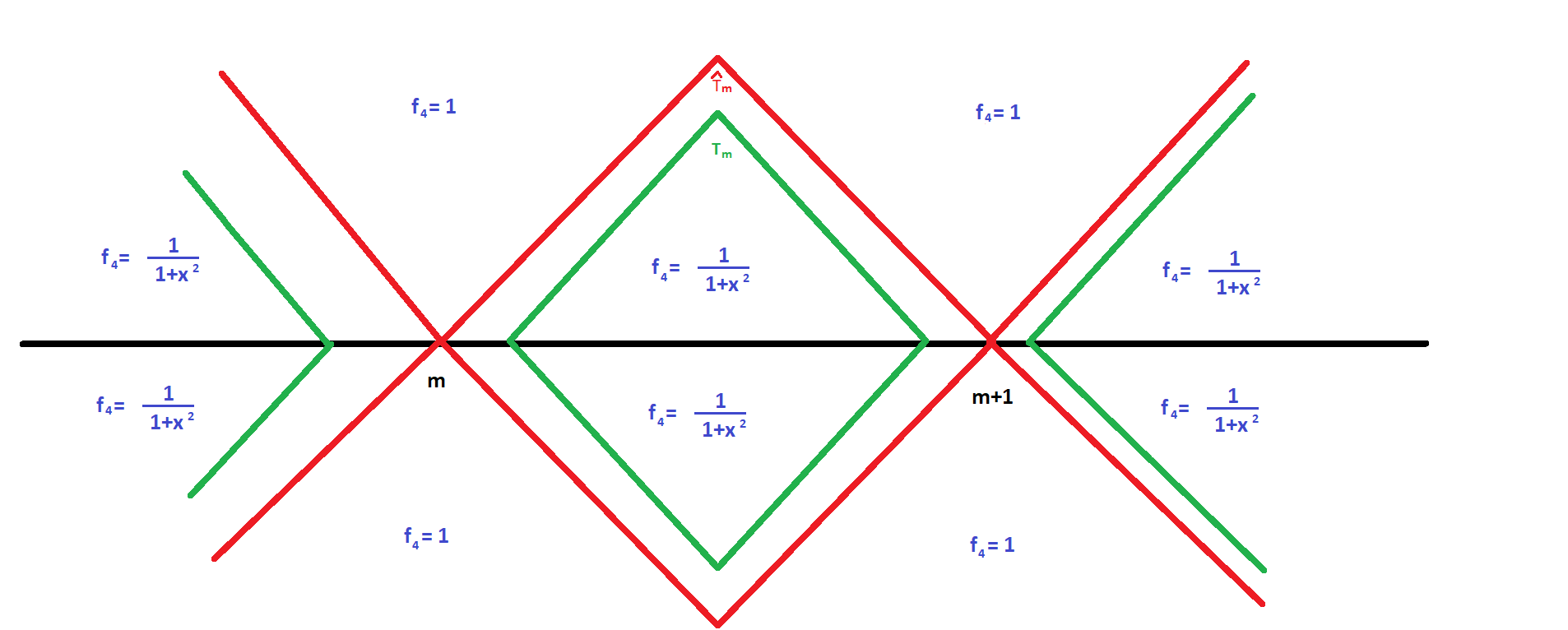}

\caption{\label{fig0} Poligons $T_m$ in green, each one inside a  polygon $\hat T_m$. Outside the region delimited by the  $\hat T_m$'s, the metric determined by $f_4$ in \eqref{e0} is as in $\Lo^2$. However, the Cauchy slice $t=0$ inherits the metric $dx^2/(1+x^2)^2$ on the $T_m$'s. This fact and the decreasing contribution for its length $\ell$ in each  $\hat T_m\setminus T_m$ makes this slice incomplete. }
\end{figure}

\subsection{Isocausal glob. hyp. spacetimes with different c-boundaries}\label{ss4}
In this example the metric \eqref{e_g}
will be defined in $(-\infty,0)\times \R$ rather than
$\R^2$.
Consider  the following metrics
 $g_{\cc} , g_5 , g_{\ca}$:
\[
g_{\cc}=-dt^2+ 4 dx^2,\qquad g_5=-dt^{2}+f_5^2\left(\left| \frac{x}{t}\right|\right)dx^{2},\qquad
g_{\ca}=-dt^2+\;dx^2,
\]
where $f_5:[0,\infty) \rightarrow [1/2,1]$ is a smooth function which
satisfies:

\begin{figure}
\includegraphics[height=0.5\textheight]{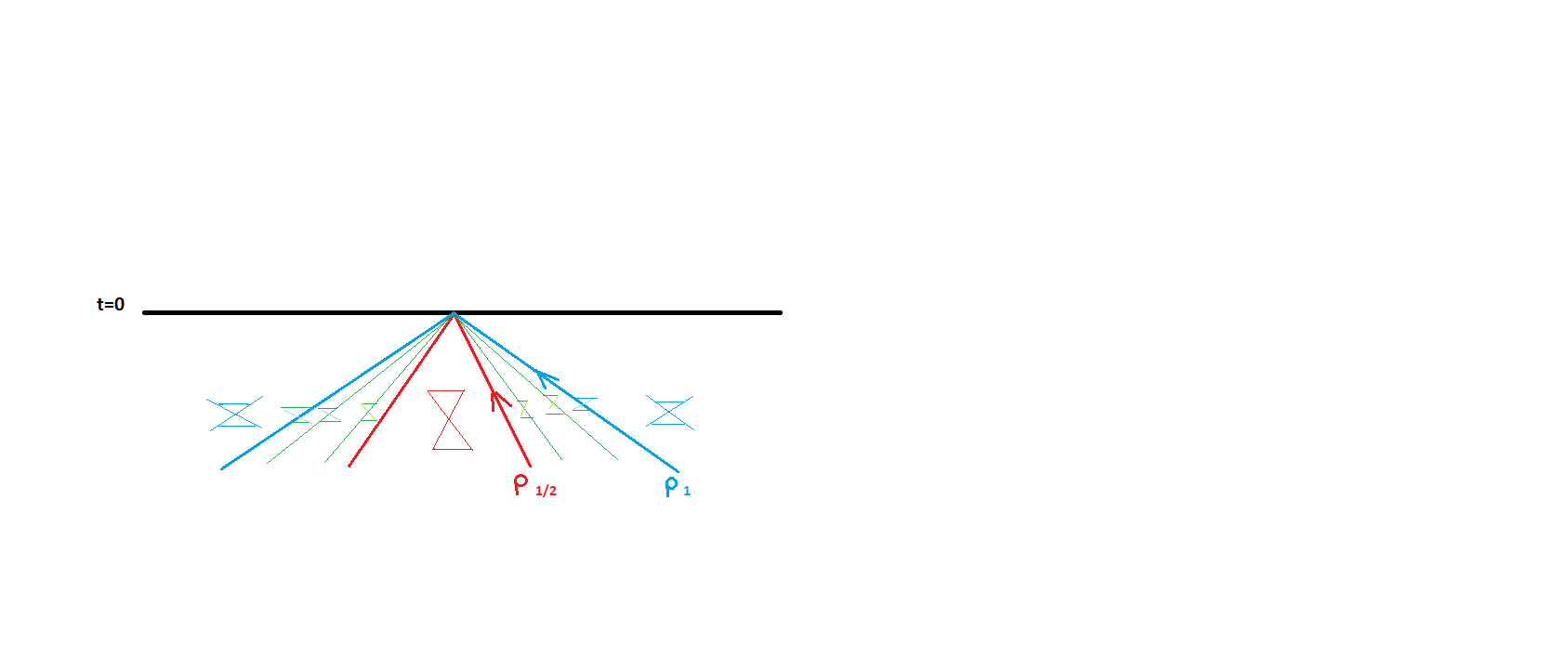}

\caption{\label{fig4} 
For the metric $g_5$, the cones in the central red wedge $|x|\leq |t|/2$ are narrower than in the lateral ones $|t|\leq |x|$ (between the blue line and the semiaxes). The green cones show the region of transition (in this region, the picture suggests $f_5(|x/t|)=|t/x|$ for simplicity, but  the metric $g_5$ is somewhat different to make it $C^\infty$).}
\end{figure}
\begin{figure}
\includegraphics[height=0.5\textheight]{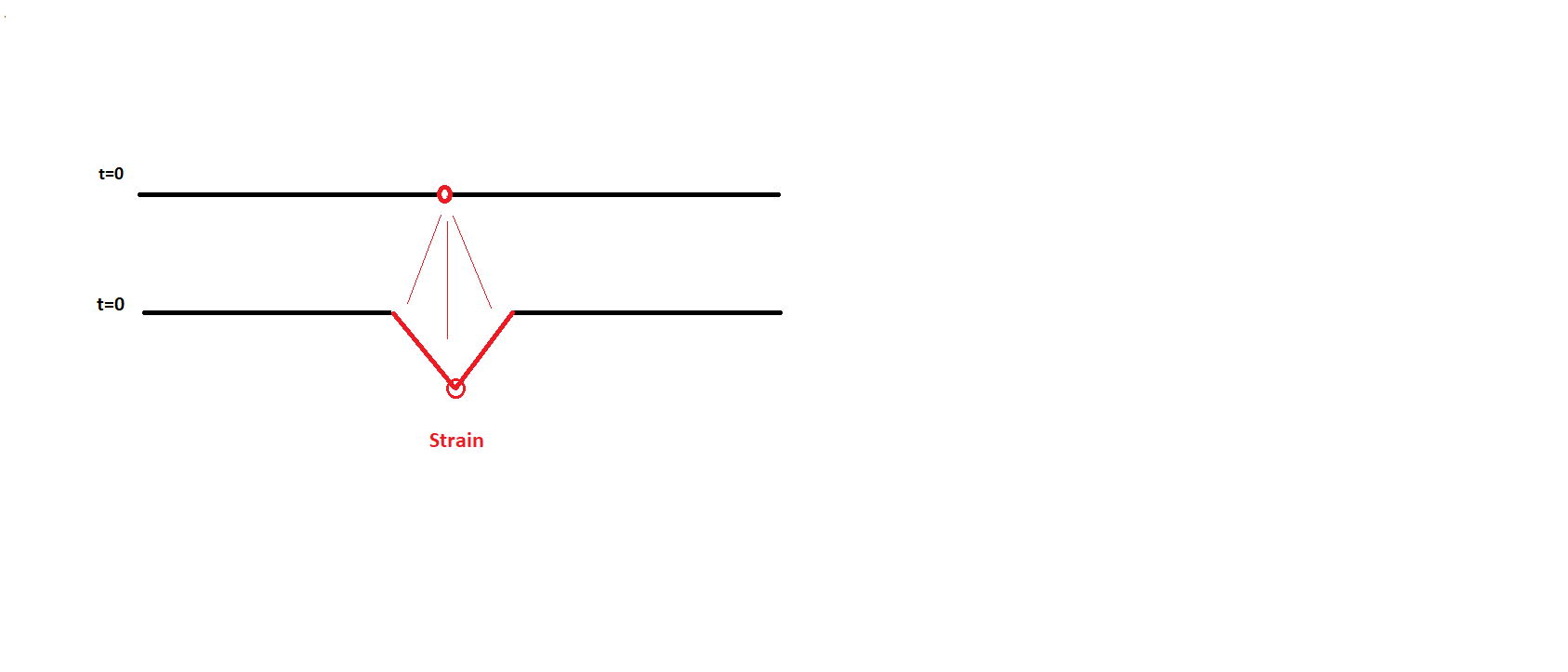}

\caption{\label{fig8} 
Above, the $x$-axis is  the future part of the c-boundary of the region $t<0$ in $\Lo^2$. Below, a strain appears in $\hat\partial\m$ for the metric $g_5$ (Fig. \ref{fig4}).}
\end{figure}

\begin{itemize}\item  $f_5(|x/t|)\equiv 2$ if $|x/t| \leq 1/2$, that
is, $g=g_{\cc}$ in the region $|x|\leq |t|/2 $.

\item $f_5$   decreases strictly from $2$ to $1$ on the interval
$ 1/2\leq |x/t|\leq 1$, 
so that the causal cones of $g_{\cc}$ (resp.
of $g$) are strictly contained in the ones of $g$ (resp. of
$g_{\ca}$) in the region $|t|/2 <|x|< |t|$.

\item $f_5(|x/t|)\equiv 1$ if $1\leq |x/t| $, that is, $g=g_{\ca}$ in
the region $|t| \leq |x|$.
\end{itemize}
Clearly,  $g_{\cc}$ and $g_{\ca}$ are isometric (under $(t,x)\mapsto (t,2x)$)
 and $g_{\cc} \prec g \prec g_{\ca}$ (as $2\geq f_5 \geq 1$); so, these three metrics are isocausal in the sense explained in \S 
\ref{s_isocausal}. 
 Moreover, the future part of the c-boundary of  $g_{\cc}$ and $g_{\ca}$ is spacelike and it can be identified with the $x$-axis (as each TIP   for $g_{\cc}$ or $g_{\ca}$ can be naturally identified with some $I^-(0,x), x\in \R$, where $I^-$ is computed with the natural extension of  $g_{\cc}$ or $g_{\ca}$ to $\R^2$, resp.). 

 However,  the causal boundary $\hat \m$  of $g_5$ contains  two lightlike segments and, so, it is not causally isomorphic to the one of  $g_{\cc}$ and $g_{\ca}$.  
 To check the appeareance of one of these segments, consider the  lightlike vector field  
 $$
 X=\partial_t - \frac{1}{f_5(|x/t|)} \partial_x 
 $$
(and analogous reasoning with $\partial_t + \partial_x/f_5$ would give the other lightlike segment), 
and take  its following integral curves $\rho_{x_0}$:
$$
\rho_{x_0}(s)=(-1+s,c_{x_0}(s)), \; \forall s\in [0,1), \quad  \rho_{x_0}(0)=(-1,x_0), \quad \hbox{for} \; \; \frac{1}{2}\leq x_0 \leq 1,
$$
where, $c_{x_0}$ is a suitable function. Clearly, 
$$
\rho_{1/2}(s)=(-1+s,\frac{1}{2}(1-s)), \quad \rho_{1}(s)=(-1+s,1-s), \quad , \quad \forall s\in [0,1), 
$$
and,  for each $s$,  $c_{x_0}(s)$ must grow strictly with $x_0 \in [1/2, 1]$  (otherwise, two integral curves of $X$ would cross), that is:
$$ c_{x_1}(s)<c_{x_2}(s), \; \forall s\in [0,1), \; \frac{1}{2} \leq x_1 < x_2 \leq 1 \qquad \lim_{s\rightarrow 1} c_{x_0}(s)= 0, \; \forall x_0 \in [\frac{1}{2},1].
$$
 Then, each $I^{-}(\rho_{x_0})$ is a TIP and it satisfies 
$$I^{-}(\rho_{x_1}) \subsetneq I^{-}(\rho_{x_2}), \quad  \hbox{whenever} \quad \frac{1}{2} \leq x_1 < x_2 \leq 1,
$$
that is, they are distinct points in the c-boundary which are causally related but non-timelike related, and the required lightlike segment of $\partial_c \m$, $\m=((-\infty,0)\times \R, g_5)$ is obtained. 
 

\begin{rem}\label{r_isocaus_paraelfuturo}  The two lightlike segments of $ \partial_c ((-\infty,0)\times \R,g_5)$ can be regarded as the strain of  $(0,0)$ (this point identified with one of $\partial_c ((-\infty,0)\times \R,g_{\cc}) \equiv \{0\}\times \R$) as in \S \ref{s_computecbound}. Indeed,   the procedure   for comparison with a static product $\R\times M$ can be extended to  our example. Here, the natural  comparison occurs replacing   $\R$ in this  product  by the interval $I=(-\infty,0)$ (whose c-boundary is trivially known). In this case, the bound of metrics \eqref{e_alpha_bound} holds with $\alpha\equiv 1/2$. Then, the first inequality in \eqref{e_alph_asymp} (i.e., the hypothesis for the comparison of the future c-boundary) should be replaced,  being $\int_{-1}^0 \left(\frac{1}{\alpha(t)}-1\right) dt <\infty$  the natural option (which is satisfied by $\alpha\equiv 1/2$).
\end{rem}

\section*{Acknowledgments}
The author acknowledges warmly the  exciting discussions and suggestions  by 
Valter Moretti, Simone Murro, Daniele Volpe, 
Felix Finster, Albert Much, Kyriakos Papadopoulos, 
Olaf M\"uller, Jonat\'an Herrera, 
 Stefan Suhr and  Leonardo Garc\'{i}a Heveling.  
Special thanks also to the organizers of the meeting SCRI21, which motivated the present contribution. The author is partially supported by
the grants P20-01391 (PAIDI 2020, Junta
de Andaluc\'{\i}a) and PID2020-116126GB-I00 (MCIN/ AEI /10.13039/501100011033), as well as the framework IMAG/  Mar\'{\i}a de Maeztu,   CEX2020-001105-MCIN/ AEI/ 10.13039/501100011033.
 

\section*{Data Availability Statement}
Data sharing is not applicable to this article as no datasets were generated or analysed during the current study.

\end{document}